\documentclass[submission,copyright,creativecommons]{eptcs}
\providecommand{\event}{QPL 2017} % Name of the event you are submitting to
\usepackage{breakurl}             % Not needed if you use pdflatex only.

\title{Double Dilation $\neq$ Double Mixing (extended abstract)}
\author{Maaike Zwart and Bob Coecke
\institute{University of Oxford \\ Department of Computer Science}
\email{maaike.zwart@cs.ox.ac.uk - bob.coecke@cs.ox.ac.uk}}
\def\titlerunning{Double Dilation $\neq$ Double Mixing}
\def\authorrunning{M. Zwart \& B. Coecke}

\usepackage[english]{babel}
\usepackage{amsmath,amsthm,amssymb}
\usepackage[svgnames]{xcolor}
\usepackage{rotating}
\usepackage{floatpag}
\usepackage{xifthen}
\usepackage{bm}
\usepackage{cite}
\usepackage{booktabs}

\usepackage{xspace,enumerate,color,epsfig}
\usepackage{graphicx}
\graphicspath{{.}{./figures/}}

\usepackage{tikzfig}
\usepackage{stmaryrd}
\usepackage{docmute}
\usepackage{keycommand}

\newkeycommand{\boxmap}[style=small box][1]{\,\tikz{\node[style=\commandkey{style}] (x) {$#1$};\draw(0,-1.25)--(x)--(0,1.25);}\,}
\newkeycommand{\dboxmap}[style=dbox][1]{\,\tikz{\node[style=\commandkey{style}] (x) {$#1$};\draw[boldedge] (0,-1.25)--(x)--(0,1.25);}\,}

\newkeycommand{\wirelabel}[style=white label]{\,\tikz{\node[style=\commandkey{style}]{};}\,\xspace}

\newkeycommand{\pointmap}[style=point][1]{\,\tikz{\node[style=\commandkey{style}] (x) at (0,-0.3) {$#1$};\node [style=none] (3) at (0, 0.9) {};\node [style=none] (3) at (0, -0.9) {};\draw (x)--(0,0.7);}\,}

\newkeycommand{\point}[style=point,estyle=][1]{\,\tikz{\node[style=\commandkey{style}] (x) at (0,-0.05) {$#1$};\draw [\commandkey{estyle}] (x)--(0,1.05);}\,}

\newkeycommand{\copointmap}[style=copoint][1]{\,\tikz{\node[style=\commandkey{style}] (x) at (0,0.3) {$#1$};\draw (x)--(0,-0.7);}\,}

\newkeycommand{\dpoint}[style=point][1]{\,\tikz{\node[style=\commandkey{style},doubled] (x) at (0,-0.3) {$#1$};\draw[boldedge] (x)--(0,0.7);}\,}

\newkeycommand{\kpoint}[style=kpoint,estyle=][1]{\,\tikz{\node[style=\commandkey{style}] (x) at (0,-0.05) {$#1$};\draw [\commandkey{estyle}] (x)--(0,1.05);}\,}

\newkeycommand{\kpointgray}[style=gray kpoint,estyle=][1]{\,\tikz{\node[style=\commandkey{style}] (x) at (0,-0.05) {$#1$};\draw [\commandkey{estyle}] (x)--(0,1.05);}\,}

\newkeycommand{\typedkpoint}[style=kpoint,edgestyle=][2]{%
\,\begin{tikzpicture}
  \begin{pgfonlayer}{nodelayer}
    \node [style=none] (0) at (0, 1) {};
    \node [style=\commandkey{style}] (1) at (0, -0.75) {$#2$};
    \node [style=right label] (2) at (0.25, 0.5) {$#1$};
  \end{pgfonlayer}
  \begin{pgfonlayer}{edgelayer}
    \draw [\commandkey{edgestyle}] (1) to (0.center);
  \end{pgfonlayer}
\end{tikzpicture}\,}

\newkeycommand{\typedkpointdag}[style=kpoint adjoint,edgestyle=][2]{%
\,\begin{tikzpicture}
  \begin{pgfonlayer}{nodelayer}
    \node [style=none] (0) at (0, -1) {};
    \node [style=\commandkey{style}] (1) at (0, 0.75) {$#2$};
    \node [style=right label] (2) at (0.25, -0.5) {$#1$};
  \end{pgfonlayer}
  \begin{pgfonlayer}{edgelayer}
    \draw [\commandkey{edgestyle}] (0.center) to (1);
  \end{pgfonlayer}
\end{tikzpicture}\,}

\newkeycommand{\bistate}[style=kpoint][1]{\,\begin{tikzpicture}
  \begin{pgfonlayer}{nodelayer}
    \node [style=\commandkey{style}, minimum width=1 cm, inner sep=2pt] (0) at (0, -0.25) {$#1$};
    \node [style=none] (1) at (-0.75, 0) {};
    \node [style=none] (2) at (0.75, 0) {};
    \node [style=none] (3) at (-0.75, 0.88) {};
    \node [style=none] (4) at (0.75, 0.88) {};
  \end{pgfonlayer}
  \begin{pgfonlayer}{edgelayer}
    \draw (1.center) to (3.center);
    \draw (2.center) to (4.center);
  \end{pgfonlayer}
\end{tikzpicture}\,}

\newkeycommand{\bistatebig}[style=kpoint][1]{\,\begin{tikzpicture}
  \begin{pgfonlayer}{nodelayer}
    \node [style=\commandkey{style}, minimum width=, minimum width=1.26 cm, inner sep=2pt] (0) at (0, -0.25) {$#1$};
    \node [style=none] (1) at (-1, 0) {};
    \node [style=none] (2) at (1, 0) {};
    \node [style=none] (3) at (-1, 0.88) {};
    \node [style=none] (4) at (1, 0.88) {};
  \end{pgfonlayer}
  \begin{pgfonlayer}{edgelayer}
    \draw (1.center) to (3.center);
    \draw (2.center) to (4.center);
  \end{pgfonlayer}
\end{tikzpicture}\,}

\newkeycommand{\dbistate}[style=dkpoint][1]{\,\begin{tikzpicture}
  \begin{pgfonlayer}{nodelayer}
    \node [style=\commandkey{style}, minimum width=1 cm, inner sep=2pt] (0) at (0, -0.25) {$#1$};
    \node [style=none] (1) at (-0.75, 0) {};
    \node [style=none] (2) at (0.75, 0) {};
    \node [style=none] (3) at (-0.75, 1.0) {};
    \node [style=none] (4) at (0.75, 1.0) {};
  \end{pgfonlayer}
  \begin{pgfonlayer}{edgelayer}
    \draw [boldedge] (1.center) to (3.center);
    \draw [boldedge] (2.center) to (4.center);
  \end{pgfonlayer}
\end{tikzpicture}\,}

\newkeycommand{\bistatebraket}[style1=kpoint,style2=kpointadj][2]{\,\begin{tikzpicture}
  \begin{pgfonlayer}{nodelayer}
    \node [style=\commandkey{style1}, minimum width=1 cm, inner sep=2pt] (0) at (0, -0.75) {$#1$};
    \node [style=none] (1) at (-0.75, -0.5) {};
    \node [style=none] (2) at (0.75, -0.5) {};
    \node [style=none] (3) at (-0.75, 0.5) {};
    \node [style=none] (4) at (0.75, 0.5) {};
    \node [style=\commandkey{style2}, minimum width=1 cm, inner sep=2pt] (5) at (0, 0.75) {$#2$};
  \end{pgfonlayer}
  \begin{pgfonlayer}{edgelayer}
    \draw (1.center) to (3.center);
    \draw (2.center) to (4.center);
  \end{pgfonlayer}
\end{tikzpicture}\,}

\newkeycommand{\weight}[style=dkpoint][1]{\,\begin{tikzpicture}
  \begin{pgfonlayer}{nodelayer}
    \node [style=\commandkey{style}] (0) at (0, -0.5) {$#1$};
    \node [style=upground] (1) at (0, 0.75) {};
  \end{pgfonlayer}
  \begin{pgfonlayer}{edgelayer}
    \draw [style=boldedge] (0) to (1);
  \end{pgfonlayer}
\end{tikzpicture}\,}

\newkeycommand{\dkpoint}[style=dkpoint][1]{\,\tikz{\node[style=\commandkey{style}] (x) at (0,-0.1) {$#1$};\draw[boldedge] (x)--(0,1);}\,}
\newkeycommand{\dkpointadj}[style=dkpointadj][1]{\,\tikz{\node[style=\commandkey{style}] (x) at (0,0.1) {$#1$};\draw[boldedge] (x)--(0,-1);}\,}
\newkeycommand{\dkpointtrans}[style=dkpointtrans][1]{\,\tikz{\node[style=\commandkey{style}] (x) at (0,0.1) {$#1$};\draw[boldedge] (x)--(0,-1);}\,}
\newkeycommand{\dkpointconj}[style=dkpointconj][1]{\,\tikz{\node[style=\commandkey{style}] (x) at (0,-0.1) {$#1$};\draw[boldedge] (x)--(0,1);}\,}

\newkeycommand{\blackkpoint}[style=black kpoint][1]{\,\tikz{\node[style=\commandkey{style}] (x) at (0,-0.1) {$#1$};\draw (x)--(0,1);}\,}
\newkeycommand{\blackkpointadj}[style=black kpointadj][1]{\,\tikz{\node[style=\commandkey{style}] (x) at (0,0.1) {$#1$};\draw (x)--(0,-1);}\,}
\newkeycommand{\blackdkpoint}[style=black dkpoint][1]{\,\tikz{\node[style=\commandkey{style}] (x) at (0,-0.1) {$#1$};\draw[boldedge] (x)--(0,1);}\,}
\newkeycommand{\blackdkpointadj}[style=black dkpointadj][1]{\,\tikz{\node[style=\commandkey{style}] (x) at (0,0.1) {$#1$};\draw[boldedge] (x)--(0,-1);}\,}

\newcommand{\trace}{\,\begin{tikzpicture}
  \begin{pgfonlayer}{nodelayer}
    \node [style=none] (0) at (0, -0.60) {};
    \node [style=upground] (1) at (0, 0.40) {};
  \end{pgfonlayer}
  \begin{pgfonlayer}{edgelayer}
    \draw [style=boldedge] (0.center) to (1);
  \end{pgfonlayer}
\end{tikzpicture}\,}

\newcommand\discard\trace

\newkeycommand{\pointketbra}[style1=copoint,style2=point,estyle=][2]{\,%
\begin{tikzpicture}
  \begin{pgfonlayer}{nodelayer}
    \node [style=none] (0) at (0, -1.75) {};
    \node [style=\commandkey{style1}] (1) at (0, -0.75) {$#1$};
    \node [style=\commandkey{style2}] (2) at (0, 0.75) {$#2$};
    \node [style=none] (3) at (0, 1.75) {};
  \end{pgfonlayer}
  \begin{pgfonlayer}{edgelayer}
    \draw [\commandkey{estyle}] (0.center) to (1);
    \draw [\commandkey{estyle}] (2) to (3.center);
  \end{pgfonlayer}
\end{tikzpicture}\,}

\newkeycommand{\twocopointketbra}[style1=copoint,style2=copoint,style3=point][3]{\,%
\begin{tikzpicture}
  \begin{pgfonlayer}{nodelayer}
    \node [style=none] (0) at (-0.75, -1.75) {};
    \node [style=none] (1) at (0.75, -1.75) {};
    \node [style=\commandkey{style1}] (2) at (-0.75, -0.75) {$#1$};
    \node [style=\commandkey{style2}] (3) at (0.75, -0.75) {$#2$};
    \node [style=\commandkey{style3}] (4) at (0, 0.75) {$#3$};
    \node [style=none] (5) at (0, 1.75) {};
  \end{pgfonlayer}
  \begin{pgfonlayer}{edgelayer}
    \draw (0.center) to (2);
    \draw (1.center) to (3);
    \draw (4) to (5.center);
  \end{pgfonlayer}
\end{tikzpicture}\,}

\newkeycommand{\twopointketbra}[style1=copoint,style2=point,style3=point][3]{\,%
\begin{tikzpicture}
  \begin{pgfonlayer}{nodelayer}
    \node [style=none] (0) at (0, -1.75) {};
    \node [style=none] (1) at (-0.75, 1.75) {};
    \node [style=\commandkey{style1}] (2) at (0, -0.75) {$#1$};
    \node [style=\commandkey{style2}] (3) at (-0.75, 0.75) {$#2$};
    \node [style=\commandkey{style3}] (4) at (0.75, 0.75) {$#3$};
    \node [style=none] (5) at (0.75, 1.75) {};
  \end{pgfonlayer}
  \begin{pgfonlayer}{edgelayer}
    \draw (0.center) to (2);
    \draw (1.center) to (3);
    \draw (4) to (5.center);
  \end{pgfonlayer}
\end{tikzpicture}\,}

\newkeycommand{\pointbraket}[style1=point,style2=copoint,estyle=][2]{\,%
\begin{tikzpicture}
  \begin{pgfonlayer}{nodelayer}
    \node [style=\commandkey{style1}] (0) at (0, -0.625) {$#1$};
    \node [style=\commandkey{style2}] (1) at (0, 0.625) {$#2$};
  \end{pgfonlayer}
  \begin{pgfonlayer}{edgelayer}
    \draw [\commandkey{estyle}] (0) to (1);
  \end{pgfonlayer}
\end{tikzpicture}\,}

\newkeycommand{\kpointketbra}[style1=kpointdag,style2=kpoint,estyle=][2]{\,%
\begin{tikzpicture}
  \begin{pgfonlayer}{nodelayer}
    \node [style=none] (0) at (0, -1.9) {};
    \node [style=\commandkey{style1}] (1) at (0, -0.95) {$#1$};
    \node [style=\commandkey{style2}] (2) at (0, 0.95) {$#2$};
    \node [style=none] (3) at (0, 1.9) {};
  \end{pgfonlayer}
  \begin{pgfonlayer}{edgelayer}
    \draw [\commandkey{estyle}] (0.center) to (1);
    \draw [\commandkey{estyle}] (2) to (3.center);
  \end{pgfonlayer}
\end{tikzpicture}\,}

\newkeycommand{\kpointmap}[style1=kpoint,style2=map][2]{\,%
\begin{tikzpicture}
  \begin{pgfonlayer}{nodelayer}
    \node [style=\commandkey{style1}] (0) at (0, -1.1) {$#1$};
    \node [style=\commandkey{style2}] (1) at (0, 0.5) {$#2$};
    \node [style=none] (2) at (0, 1.75) {};
  \end{pgfonlayer}
  \begin{pgfonlayer}{edgelayer}
    \draw (0) to (1);
    \draw (1) to (2.center);
  \end{pgfonlayer}
\end{tikzpicture}%
\,}

\newkeycommand{\kpointmaptrans}[style1=kpointconj,style2=mapadj][2]{\,%
\begin{tikzpicture}
  \begin{pgfonlayer}{nodelayer}
    \node [style=\commandkey{style1}] (0) at (0, -1.1) {$#1$};
    \node [style=\commandkey{style2}] (1) at (0, 0.5) {$#2$};
    \node [style=none] (2) at (0, 1.75) {};
  \end{pgfonlayer}
  \begin{pgfonlayer}{edgelayer}
    \draw (0) to (1);
    \draw (1) to (2.center);
  \end{pgfonlayer}
\end{tikzpicture}%
\,}

\newkeycommand{\kpointmapadj}[style1=kpointadj,style2=map][2]{\,%
\begin{tikzpicture}
  \begin{pgfonlayer}{nodelayer}
    \node [style=\commandkey{style1}] (0) at (0, 1.1) {$#1$};
    \node [style=\commandkey{style2}] (1) at (0, -0.5) {$#2$};
    \node [style=none] (2) at (0, -1.75) {};
  \end{pgfonlayer}
  \begin{pgfonlayer}{edgelayer}
    \draw (0) to (1);
    \draw (1) to (2.center);
  \end{pgfonlayer}
\end{tikzpicture}%
\,}

\newkeycommand{\dkpointmap}[style1=dkpoint,style2=dmap][2]{\,%
\begin{tikzpicture}
  \begin{pgfonlayer}{nodelayer}
    \node [style=\commandkey{style1}] (0) at (0, -1.2) {$#1$};
    \node [style=\commandkey{style2}] (1) at (0, 0.4) {$#2$};
    \node [style=none] (2) at (0, 1.7) {};
  \end{pgfonlayer}
  \begin{pgfonlayer}{edgelayer}
    \draw[style=boldedge] (0) to (1);
    \draw[style=boldedge] (1) to (2.center);
  \end{pgfonlayer}
\end{tikzpicture}%
\,}

\newkeycommand{\dkpointmapx}[style1=dkpoint,style2=dmap][2]{\,%
\begin{tikzpicture}
  \begin{pgfonlayer}{nodelayer}
    \node [style=\commandkey{style1}] (0) at (0, -1.4) {$#1$};
    \node [style=\commandkey{style2}] (1) at (0, 0.6) {$#2$};
    \node [style=none] (2) at (0, 1.9) {};
  \end{pgfonlayer}
  \begin{pgfonlayer}{edgelayer}
    \draw[style=boldedge] (0) to (1);
    \draw[style=boldedge] (1) to (2.center);
  \end{pgfonlayer}
\end{tikzpicture}%
\,}

\newkeycommand{\boxcopointmapdag}[style1=map,style2=copoint][2]{\,%
\begin{tikzpicture}
  \begin{pgfonlayer}{nodelayer}
    \node [style=none] (0) at (0, -1.75) {};
    \node [style=\commandkey{style1}] (1) at (0, -0.5) {$#1$};
    \node [style=\commandkey{style2}] (2) at (0, 1) {$#2$};
  \end{pgfonlayer}
  \begin{pgfonlayer}{edgelayer}
    \draw (0.center) to (1);
    \draw (1) to (2);
  \end{pgfonlayer}
\end{tikzpicture}%
\,}

\newkeycommand{\kpointdagmap}[style1=map,style2=kpointdag][2]{\,%
\begin{tikzpicture}
  \begin{pgfonlayer}{nodelayer}
    \node [style=none] (0) at (0, -1.75) {};
    \node [style=\commandkey{style1}] (1) at (0, -0.5) {$#1$};
    \node [style=\commandkey{style2}] (2) at (0, 1.1) {$#2$};
  \end{pgfonlayer}
  \begin{pgfonlayer}{edgelayer}
    \draw (0.center) to (1);
    \draw (1) to (2);
  \end{pgfonlayer}
\end{tikzpicture}%
\,}

\newkeycommand{\sandwichmap}[style1=point,style2=map,style3=copoint][3]{\,%
\begin{tikzpicture}
  \begin{pgfonlayer}{nodelayer}
    \node [style=\commandkey{style1}] (0) at (0, -1.50) {$#1$};
    \node [style=\commandkey{style2}] (1) at (0, 0) {$#2$};
    \node [style=\commandkey{style3}] (2) at (0, 1.50) {$#3$};
  \end{pgfonlayer}
  \begin{pgfonlayer}{edgelayer}
    \draw (0) to (1);
    \draw (1) to (2);
  \end{pgfonlayer}
\end{tikzpicture}%
}

\newkeycommand{\kpointsandwichmap}[style1=kpoint,style2=map,style3=kpointdag][3]{\,%
\begin{tikzpicture}
  \begin{pgfonlayer}{nodelayer}
    \node [style=\commandkey{style1}] (0) at (0, -1.5) {$#1$};
    \node [style=\commandkey{style2}] (1) at (0, 0) {$#2$};
    \node [style=\commandkey{style3}] (2) at (0, 1.5) {$#3$};
  \end{pgfonlayer}
  \begin{pgfonlayer}{edgelayer}
    \draw (0) to (1);
    \draw (1) to (2);
  \end{pgfonlayer}
\end{tikzpicture}%
}

\newkeycommand{\sandwichmapconj}[style1=point,style2=mapconj,style3=copoint][3]{\,% 
\begin{tikzpicture}
  \begin{pgfonlayer}{nodelayer}
    \node [style=\commandkey{style1}] (0) at (0, -1.50) {$#1$};
    \node [style=\commandkey{style2}] (1) at (0, 0) {$#2$};
    \node [style=\commandkey{style3}] (2) at (0, 1.50) {$#3$};
  \end{pgfonlayer}
  \begin{pgfonlayer}{edgelayer}
    \draw (0) to (1);
    \draw (1) to (2);
  \end{pgfonlayer}
\end{tikzpicture}%
}

\newkeycommand{\sandwichtwo}[style1=point,style2=map,style3=map,style4=copoint][4]{\,%
\begin{tikzpicture}
  \begin{pgfonlayer}{nodelayer}
    \node [style=\commandkey{style2}] (0) at (0, -1) {$#2$};
    \node [style=\commandkey{style3}] (1) at (0, 1) {$#3$};
    \node [style=\commandkey{style1}] (2) at (0, -2.6) {$#1$};
    \node [style=\commandkey{style4}] (3) at (0, 2.6) {$#4$};
  \end{pgfonlayer}
  \begin{pgfonlayer}{edgelayer}
    \draw (2) to (0);
    \draw (0) to (1);
    \draw (1) to (3);
  \end{pgfonlayer}
\end{tikzpicture}\,}

\newkeycommand{\longbraket}[style1=point,style2=copoint][2]{\,% 
\begin{tikzpicture}
  \begin{pgfonlayer}{nodelayer}
    \node [style=\commandkey{style1}] (0) at (0, -2.6) {$#1$};
    \node [style=\commandkey{style2}] (1) at (0, 2.6) {$#2$};
  \end{pgfonlayer}
  \begin{pgfonlayer}{edgelayer}
    \draw (0) to (1);
  \end{pgfonlayer}
\end{tikzpicture}\,}

\newkeycommand{\onb}[style=point]{\ensuremath{\{\,\tikz{\node[style=\commandkey{style}] (x) at (0,-0.3) {$j$};\draw (x)--(0,0.7);}\,\}}\xspace}

\newkeycommand{\phase}[style=white phase dot][1]{\,\begin{tikzpicture}
    \begin{pgfonlayer}{nodelayer}
        \node [style=none] (0) at (0, 1) {};
        \node [style=\commandkey{style}] (2) at (0, -0) {$#1$}; 
        \node [style=none] (3) at (0, -1) {};
    \end{pgfonlayer}
    \begin{pgfonlayer}{edgelayer}
        \draw (2) to (0.center);
        \draw (3.center) to (2);
    \end{pgfonlayer}
\end{tikzpicture}\,}

\newkeycommand{\dphase}[style=white phase ddot][1]{\,\begin{tikzpicture}
    \begin{pgfonlayer}{nodelayer}
        \node [style=none] (0) at (0, 1.25) {};
        \node [style=\commandkey{style}] (2) at (0, -0) {$#1$};
        \node [style=none] (3) at (0, -1.25) {};
    \end{pgfonlayer}
    \begin{pgfonlayer}{edgelayer}
        \draw [boldedge] (2) to (0.center);
        \draw [boldedge] (3.center) to (2);
    \end{pgfonlayer}
\end{tikzpicture}\,}

\newkeycommand{\dphasepoint}[style=white phase ddot][1]{\,\begin{tikzpicture}[yshift=-3mm]
  \begin{pgfonlayer}{nodelayer}
    \node [style=none] (0) at (0, 1) {};
    \node [style=\commandkey{style}] (1) at (0, 0) {$#1$};
  \end{pgfonlayer}
  \begin{pgfonlayer}{edgelayer}
    \draw [style=boldedge] (0.center) to (1);
  \end{pgfonlayer}
\end{tikzpicture}\,}

\newkeycommand{\dphasepointgray}[style=gray phase ddot][1]{\,\begin{tikzpicture}[yshift=-3mm]
  \begin{pgfonlayer}{nodelayer}
    \node [style=none] (0) at (0, 1) {};
    \node [style=\commandkey{style}] (1) at (0, 0) {$#1$};
  \end{pgfonlayer}
  \begin{pgfonlayer}{edgelayer}
    \draw [style=boldedge] (0.center) to (1);
  \end{pgfonlayer}
\end{tikzpicture}\,}

\newkeycommand{\dphasecopoint}[style=white phase ddot][1]{\,\begin{tikzpicture}[yshift=3mm,yscale=-1]
  \begin{pgfonlayer}{nodelayer}
    \node [style=none] (0) at (0, 1) {};
    \node [style=\commandkey{style}] (1) at (0, 0) {$#1$};
  \end{pgfonlayer}
  \begin{pgfonlayer}{edgelayer}
    \draw [style=boldedge] (0.center) to (1);
  \end{pgfonlayer}
\end{tikzpicture}\,}

\newkeycommand{\dphasepointsm}[style=white phase ddot][1]{\,\begin{tikzpicture}[yshift=-3mm]
  \begin{pgfonlayer}{nodelayer}
    \node [style=none] (0) at (0, 1) {};
    \node [style=\commandkey{style}] (1) at (0, 0) {\footnotesize\!$#1$\!};
  \end{pgfonlayer}
  \begin{pgfonlayer}{edgelayer}
    \draw [style=boldedge] (0.center) to (1);
  \end{pgfonlayer}
\end{tikzpicture}\,}

\newkeycommand{\phasepoint}[style=white phase dot][1]{\,\begin{tikzpicture}[yshift=-3mm]
  \begin{pgfonlayer}{nodelayer}
    \node [style=none] (0) at (0, 1) {};
    \node [style=\commandkey{style}] (1) at (0, 0) {$#1$};
  \end{pgfonlayer}
  \begin{pgfonlayer}{edgelayer}
    \draw (0.center) to (1);
  \end{pgfonlayer}
\end{tikzpicture}\,}

\newkeycommand{\phasecopoint}[style=white phase dot][1]{\,\begin{tikzpicture}[yshift=3mm]
  \begin{pgfonlayer}{nodelayer}
    \node [style=none] (0) at (0, -1) {};
    \node [style=\commandkey{style}] (1) at (0, 0) {$#1$};
  \end{pgfonlayer}
  \begin{pgfonlayer}{edgelayer}
    \draw (0.center) to (1);
  \end{pgfonlayer}
\end{tikzpicture}\,}

\newkeycommand{\kpointbraket}[style1=kpoint,style2=kpoint adjoint,estyle=][2]{\,%
\begin{tikzpicture}
  \begin{pgfonlayer}{nodelayer}
    \node [style=\commandkey{style1}] (0) at (0, -0.735) {$#1$};
    \node [style=\commandkey{style2}] (1) at (0, 0.735) {$#2$};
  \end{pgfonlayer}
  \begin{pgfonlayer}{edgelayer}
    \draw [\commandkey{estyle}] (0) to (1);
  \end{pgfonlayer}
\end{tikzpicture}\,}

\newkeycommand{\kkpointbraket}[style1=kpoint,style2=kpoint adjoint,estyle=][2]{\,%
\begin{tikzpicture}
  \begin{pgfonlayer}{nodelayer}
    \node [style=\commandkey{style1}] (0) at (0, -0.685) {$#1$};
    \node [style=\commandkey{style2}] (1) at (0, 0.685) {$#2$};
  \end{pgfonlayer}
  \begin{pgfonlayer}{edgelayer}
    \draw [\commandkey{estyle}] (0) to (1);
  \end{pgfonlayer}
\end{tikzpicture}\,}

\newkeycommand{\kpointbraketx}[style1=kpoint,style2=kpoint adjoint,estyle=][2]{\,%
\begin{tikzpicture}
  \begin{pgfonlayer}{nodelayer}
    \node [style=\commandkey{style1}] (0) at (0, -0.885) {$#1$};
    \node [style=\commandkey{style2}] (1) at (0, 0.885) {$#2$};
  \end{pgfonlayer}
  \begin{pgfonlayer}{edgelayer}
    \draw [\commandkey{estyle}] (0) to (1);
  \end{pgfonlayer}
\end{tikzpicture}\,}

\tikzstyle{cdiag}=[matrix of math nodes, row sep=3em, column sep=3em, text height=1.5ex, text depth=0.25ex,inner sep=0.5em]
\tikzstyle{arrow above}=[transform canvas={yshift=0.5ex}]
\tikzstyle{arrow below}=[transform canvas={yshift=-0.5ex}]

\usepackage[svgnames]{xcolor}
\usepackage{tikz}
\usetikzlibrary{decorations.markings}
\usetikzlibrary{shapes.geometric}
\pgfdeclarelayer{edgelayer}
\pgfdeclarelayer{nodelayer}
\pgfsetlayers{edgelayer,nodelayer,main}

\tikzstyle{none}=[inner sep=0pt]
\definecolor{hexcolor0xff0000}{rgb}{1.000,0.000,0.000}
\definecolor{hexcolor0x000000}{rgb}{0.000,0.000,0.000}
\definecolor{hexcolor0x00ff00}{rgb}{0.000,1.000,0.000}
\definecolor{hexcolor0x000000}{rgb}{0.000,0.000,0.000}
\definecolor{hexcolor0xffff00}{rgb}{1.000,1.000,0.000}
\definecolor{hexcolor0xffffff}{rgb}{1.000,1.000,1.000}

\tikzstyle{rn}=[circle,fill=hexcolor0xff0000,draw=hexcolor0x000000,line width=0.8 pt]
\tikzstyle{gn}=[circle,fill=hexcolor0x00ff00,draw=hexcolor0x000000,line width=0.8 pt]
\tikzstyle{yn}=[circle,fill=hexcolor0xffff00,draw=hexcolor0x000000,line width=0.8 pt]
\tikzstyle{wn}=[circle,fill=hexcolor0xffffff,draw=hexcolor0x000000,line width=0.8 pt]
\tikzstyle{wnthick}=[circle,fill=hexcolor0xffffff,draw=hexcolor0x000000,line width=2.500]

\tikzstyle{simple}=[-,draw=hexcolor0x000000,line width=2.000]
\tikzstyle{arrow}=[-,draw=hexcolor0x000000,postaction={decorate},decoration={markings,mark=at position .5 with {\arrow{>}}},line width=2.000]
\tikzstyle{tick}=[-,draw=hexcolor0x000000,postaction={decorate},decoration={markings,mark=at position .5 with {\draw (0,-0.1) -- (0,0.1);}},line width=2.000]
\tikzstyle{halfthickness}=[-,draw=hexcolor0x000000,line width=0.500]
\tikzstyle{thick}=[-,draw=hexcolor0x000000,line width=2.500]
\tikzstyle{thicker}=[-,draw=hexcolor0x000000,line width=4.000]

\tikzstyle{env}=[copoint,regular polygon rotate=0,minimum width=0.2cm, fill=black]

\tikzstyle{probs}=[shape=semicircle,fill=white,draw=black,shape border rotate=180,minimum width=1.2cm]

\tikzstyle{every picture}=[baseline=-0.25em,scale=0.5]
\tikzstyle{dotpic}=[] % for backwards-compatibility
\tikzstyle{diredges}=[every to/.style={diredge}]
\tikzstyle{math matrix}=[matrix of math nodes,left delimiter=(,right delimiter=),inner sep=2pt,column sep=1em,row sep=0.5em,nodes={inner sep=0pt},text height=1.5ex, text depth=0.25ex]

\tikzstyle{inline text}=[text height=1.5ex, text depth=0.25ex,yshift=0.5mm]
\tikzstyle{label}=[font=\footnotesize,text height=1.5ex, text depth=0.25ex,yshift=0.5mm]
\tikzstyle{left label}=[label,anchor=east,xshift=1.5mm]
\tikzstyle{right label}=[label,anchor=west,xshift=-1.5mm]

\tikzstyle{braceedge}=[decorate,decoration={brace,amplitude=2mm,raise=-1mm}]
\tikzstyle{small braceedge}=[decorate,decoration={brace,amplitude=1mm,raise=-1mm}]

\tikzstyle{doubled}=[line width=1.6pt] % set the line width for all doubled (quantum) maps/wires
\tikzstyle{boldedge}=[doubled,shorten <=-0.17mm,shorten >=-0.17mm]
\tikzstyle{boldedgegray}=[doubled,gray,shorten <=-0.17mm,shorten >=-0.17mm]

\tikzstyle{semidoubled}=[line width=1.4pt] % set the line width for all doubled (quantum) maps/wires
\tikzstyle{semiboldedgegray}=[semidoubled,gray,shorten <=-0.17mm,shorten >=-0.17mm]

\tikzstyle{boldedgedashed}=[very thick,dashed,shorten <=-0.17mm,shorten >=-0.17mm]
\tikzstyle{vboldedgedashed}=[doubled,dashed,shorten <=-0.17mm,shorten >=-0.17mm]
\tikzstyle{left hook arrow}=[left hook-latex]
\tikzstyle{right hook arrow}=[right hook-latex]
\tikzstyle{sembracket}=[line width=0.5pt,shorten <=-0.07mm,shorten >=-0.07mm]

\tikzstyle{causal edge}=[->,thick,gray]
\tikzstyle{causal nondir}=[thick,gray]
\tikzstyle{timeline}=[thick,gray, dashed]

\tikzstyle{cedge}=[<->,thick,gray!70!white]

\tikzstyle{empty diagram}=[draw=gray!40!white,dashed,shape=rectangle,minimum width=1cm,minimum height=1cm]
\tikzstyle{empty diagram small}=[draw=gray!50!white,dashed,shape=rectangle,minimum width=0.6cm,minimum height=0.5cm]

\tikzstyle{dot}=[inner sep=0mm,minimum width=2mm,minimum height=2mm,draw,shape=circle]
\tikzstyle{ddot}=[inner sep=0mm, doubled, minimum width=2.5mm,minimum height=2.5mm,draw,shape=circle]

\tikzstyle{black dot}=[dot,fill=black]
\tikzstyle{white dot}=[dot,fill=white,,text depth=-0.2mm]
\tikzstyle{green dot}=[white dot] % for backwards-compatibility
\tikzstyle{gray dot}=[dot,fill=gray!40!white,,text depth=-0.2mm]
\tikzstyle{red dot}=[gray dot] % for backwards-compatibility

\tikzstyle{black ddot}=[ddot,fill=black]
\tikzstyle{white ddot}=[ddot,fill=white]
\tikzstyle{gray ddot}=[ddot,fill=gray!40!white]

\tikzstyle{gray edge}=[gray!40!white]

\tikzstyle{small dot}=[inner sep=0.5mm,minimum width=0pt,minimum height=0pt,draw,shape=circle]

\tikzstyle{small black dot}=[small dot,fill=black]
\tikzstyle{small white dot}=[small dot,fill=white]
\tikzstyle{small gray dot}=[small dot,fill=gray!40!white]

\tikzstyle{causal dot}=[inner sep=0.4mm,minimum width=0pt,minimum height=0pt,draw=white,shape=circle,fill=gray!40!white]

\tikzstyle{phase dimensions}=[minimum size=5mm,font=\footnotesize,rectangle,rounded corners=2.5mm,inner sep=0.2mm,outer sep=-2mm]
\tikzstyle{dphase dimensions}=[minimum size=5mm,font=\footnotesize,rectangle,rounded corners=2.5mm,inner sep=0.2mm,outer sep=-2mm]
\tikzstyle{white phase dot}=[dot,fill=white,phase dimensions]
\tikzstyle{white phase ddot}=[ddot,fill=white,dphase dimensions]

\tikzstyle{white rect ddot}=[draw=black,fill=white,doubled,minimum size=5mm,font=\footnotesize,rectangle,rounded corners=2.5mm,inner sep=0.2mm]
\tikzstyle{gray rect ddot}=[draw=black,fill=gray!40!white,doubled,minimum size=6mm,font=\footnotesize,rectangle,rounded corners=3mm]

\tikzstyle{gray phase dot}=[dot,fill=gray!40!white,phase dimensions]
\tikzstyle{gray phase ddot}=[ddot,fill=gray!40!white,dphase dimensions]
\tikzstyle{grey phase dot}=[gray phase dot]
\tikzstyle{grey phase ddot}=[gray phase ddot]

\tikzstyle{small phase dimensions}=[minimum size=4mm,font=\tiny,rectangle,rounded corners=2mm,inner sep=0.2mm,outer sep=-2mm]
\tikzstyle{small dphase dimensions}=[minimum size=4mm,font=\tiny,rectangle,rounded corners=2mm,inner sep=0.2mm,outer sep=-2mm]

\tikzstyle{small gray phase dot}=[dot,fill=gray!40!white,small phase dimensions]
\tikzstyle{small gray phase ddot}=[ddot,fill=gray!40!white,small dphase dimensions]

\tikzstyle{small map}=[draw,shape=rectangle,minimum height=4mm,minimum width=4mm,fill=white]

\tikzstyle{cnot}=[fill=white,shape=circle,inner sep=-1.4pt]

\tikzstyle{asym hadamard}=[fill=white,draw,shape=NEbox,inner sep=0.6mm,font=\footnotesize,minimum height=4mm]
\tikzstyle{asym hadamard conj}=[fill=white,draw,shape=NWbox,inner sep=0.6mm,font=\footnotesize,minimum height=4mm]
\tikzstyle{asym hadamard dag}=[fill=white,draw,shape=SEbox,inner sep=0.6mm,font=\footnotesize,minimum height=4mm]

\tikzstyle{hadamard}=[fill=white,draw,inner sep=0.6mm,font=\footnotesize,minimum height=4mm,minimum width=4mm]
\tikzstyle{small hadamard}=[fill=white,draw,inner sep=0.6mm,minimum height=1.5mm,minimum width=1.5mm]
\tikzstyle{dhadamard}=[hadamard,doubled]
\tikzstyle{small dhadamard}=[small hadamard,doubled]
\tikzstyle{small dhadamard rotate}=[small hadamard,doubled,rotate=45]
\tikzstyle{antipode}=[white dot,inner sep=0.3mm,font=\footnotesize]

\tikzstyle{scalar}=[diamond,draw,inner sep=0.5pt,font=\small]
\tikzstyle{dscalar}=[diamond,doubled, draw,inner sep=0.5pt,font=\small]

\tikzstyle{small box}=[rectangle,inline text,fill=white,draw,minimum height=5mm,yshift=-0.5mm,minimum width=5mm,font=\small]
\tikzstyle{small gray box}=[small box,fill=gray!30]
\tikzstyle{medium box}=[rectangle,inline text,fill=white,draw,minimum height=5mm,yshift=-0.5mm,minimum width=10mm,font=\small]
\tikzstyle{square box}=[small box] % for backwards-compatibility
\tikzstyle{medium gray box}=[small box,fill=gray!30]
\tikzstyle{semilarge box}=[rectangle,inline text,fill=white,draw,minimum height=5mm,yshift=-0.5mm,minimum width=12.5mm,font=\small]
\tikzstyle{large box}=[rectangle,inline text,fill=white,draw,minimum height=5mm,yshift=-0.5mm,minimum width=15mm,font=\small]
\tikzstyle{large gray box}=[small box,fill=gray!30]

\tikzstyle{Bayes box}=[rectangle,fill=black,draw, minimum height=3mm, minimum width=3mm]

\tikzstyle{gray square point}=[small box,fill=gray!50]

\tikzstyle{dphase box white}=[dhadamard]
\tikzstyle{dphase box gray}=[dhadamard,fill=gray!50!white]

\tikzstyle{point}=[regular polygon,regular polygon sides=3,draw,scale=0.75,inner sep=-0.5pt,minimum width=9mm,fill=white,regular polygon rotate=180]
\tikzstyle{copoint}=[regular polygon,regular polygon sides=3,draw,scale=0.75,inner sep=-0.5pt,minimum width=9mm,fill=white]
\tikzstyle{dpoint}=[point,doubled]
\tikzstyle{dcopoint}=[copoint,doubled]

\tikzstyle{wide copoint}=[fill=white,draw,shape=isosceles triangle,shape border rotate=90,isosceles triangle stretches=true,inner sep=0pt,minimum width=1.5cm,minimum height=6.12mm]
\tikzstyle{wide point}=[fill=white,draw,shape=isosceles triangle,shape border rotate=-90,isosceles triangle stretches=true,inner sep=0pt,minimum width=1.5cm,minimum height=6.12mm,yshift=-0.0mm]
\tikzstyle{wide point plus}=[fill=white,draw,shape=isosceles triangle,shape border rotate=-90,isosceles triangle stretches=true,inner sep=0pt,minimum width=1.74cm,minimum height=7mm,yshift=-0.0mm]

\tikzstyle{wide dpoint}=[fill=white,doubled,draw,shape=isosceles triangle,shape border rotate=-90,isosceles triangle stretches=true,inner sep=0pt,minimum width=1.5cm,minimum height=6.12mm,yshift=-0.0mm]
\tikzstyle{wide dcopoint}=[fill=white,doubled,draw,shape=isosceles triangle,shape border rotate=90,isosceles triangle stretches=true,inner sep=0pt,minimum width=1.5cm,minimum height=6.12mm,yshift=-0.0mm]

\tikzstyle{tinypoint}=[regular polygon,regular polygon sides=3,draw,scale=0.55,inner sep=-0.15pt,minimum width=6mm,fill=white,regular polygon rotate=180]

\tikzstyle{white point}=[point]
\tikzstyle{white dpoint}=[dpoint]
\tikzstyle{green point}=[white point] % for backwards-compatibility
\tikzstyle{white copoint}=[copoint]
\tikzstyle{gray point}=[point,fill=gray!40!white]
\tikzstyle{gray dpoint}=[gray point,doubled]
\tikzstyle{red point}=[gray point] % for backwards-compatibility
\tikzstyle{gray copoint}=[copoint,fill=gray!40!white]
\tikzstyle{gray dcopoint}=[gray copoint,doubled]

\tikzstyle{white point guide}=[regular polygon,regular polygon sides=3,font=\scriptsize,draw,scale=0.65,inner sep=-0.5pt,minimum width=9mm,fill=white,regular polygon rotate=180]

\tikzstyle{black point}=[point,fill=black,font=\color{white}]
\tikzstyle{black copoint}=[copoint,fill=black,font=\color{white}]

\tikzstyle{tiny gray point}=[tinypoint,fill=gray!40!white]

\tikzstyle{diredge}=[->]
\tikzstyle{ddiredge}=[<->]
\tikzstyle{rdiredge}=[<-]
\tikzstyle{thickdiredge}=[->, very thick]
\tikzstyle{pointer edge}=[->,very thick,gray]
\tikzstyle{pointer edge part}=[very thick,gray]
\tikzstyle{dashed edge}=[dashed]
\tikzstyle{thick dashed edge}=[very thick,dashed]
\tikzstyle{thick gray dashed edge}=[thick dashed edge,gray!40]
\tikzstyle{thick map edge}=[very thick,|->]

\makeatletter
\newcommand{\boxshape}[3]{%
\pgfdeclareshape{#1}{
\inheritsavedanchors[from=rectangle] % this is nearly a rectangle
\inheritanchorborder[from=rectangle]
\inheritanchor[from=rectangle]{center}
\inheritanchor[from=rectangle]{north}
\inheritanchor[from=rectangle]{south}
\inheritanchor[from=rectangle]{west}
\inheritanchor[from=rectangle]{east}
\backgroundpath{% this is new
\southwest \pgf@xa=\pgf@x \pgf@ya=\pgf@y
\northeast \pgf@xb=\pgf@x \pgf@yb=\pgf@y

\@tempdima=#2
\@tempdimb=#3

\pgfpathmoveto{\pgfpoint{\pgf@xa - 5pt + \@tempdima}{\pgf@ya}}
\pgfpathlineto{\pgfpoint{\pgf@xa - 5pt - \@tempdima}{\pgf@yb}}
\pgfpathlineto{\pgfpoint{\pgf@xb + 5pt + \@tempdimb}{\pgf@yb}}
\pgfpathlineto{\pgfpoint{\pgf@xb + 5pt - \@tempdimb}{\pgf@ya}}
\pgfpathlineto{\pgfpoint{\pgf@xa - 5pt + \@tempdima}{\pgf@ya}}
\pgfpathclose
}
}}

\boxshape{NEbox}{0pt}{5pt}
\boxshape{SEbox}{0pt}{-5pt}
\boxshape{NWbox}{5pt}{0pt}
\boxshape{SWbox}{-5pt}{0pt}
\boxshape{EBox}{-3pt}{3pt}
\boxshape{WBox}{3pt}{-3pt}
\makeatother

\tikzstyle{cloud}=[shape=cloud,draw,minimum width=1.5cm,minimum height=1.5cm]

\tikzstyle{map}=[draw,shape=NEbox,inner sep=2pt,minimum height=6mm,fill=white]
\tikzstyle{dashedmap}=[draw,dashed,shape=NEbox,inner sep=2pt,minimum height=6mm,fill=white]
\tikzstyle{mapdag}=[draw,shape=SEbox,inner sep=2pt,minimum height=6mm,fill=white]
\tikzstyle{mapadj}=[draw,shape=SEbox,inner sep=2pt,minimum height=6mm,fill=white]
\tikzstyle{maptrans}=[draw,shape=SWbox,inner sep=2pt,minimum height=6mm,fill=white]
\tikzstyle{mapconj}=[draw,shape=NWbox,inner sep=2pt,minimum height=6mm,fill=white]

\tikzstyle{medium map}=[draw,shape=NEbox,inner sep=2pt,minimum height=6mm,fill=white,minimum width=7mm]
\tikzstyle{medium map dag}=[draw,shape=SEbox,inner sep=2pt,minimum height=6mm,fill=white,minimum width=7mm]
\tikzstyle{medium map adj}=[draw,shape=SEbox,inner sep=2pt,minimum height=6mm,fill=white,minimum width=7mm]
\tikzstyle{medium map trans}=[draw,shape=SWbox,inner sep=2pt,minimum height=6mm,fill=white,minimum width=7mm]
\tikzstyle{medium map conj}=[draw,shape=NWbox,inner sep=2pt,minimum height=6mm,fill=white,minimum width=7mm]
\tikzstyle{semilarge map}=[draw,shape=NEbox,inner sep=2pt,minimum height=6mm,fill=white,minimum width=9.5mm]
\tikzstyle{semilarge map trans}=[draw,shape=SWbox,inner sep=2pt,minimum height=6mm,fill=white,minimum width=9.5mm]
\tikzstyle{semilarge map adj}=[draw,shape=SEbox,inner sep=2pt,minimum height=6mm,fill=white,minimum width=9.5mm]
\tikzstyle{semilarge map dag}=[draw,shape=SEbox,inner sep=2pt,minimum height=6mm,fill=white,minimum width=9.5mm]
\tikzstyle{semilarge map conj}=[draw,shape=NWbox,inner sep=2pt,minimum height=6mm,fill=white,minimum width=9.5mm]
\tikzstyle{large map}=[draw,shape=NEbox,inner sep=2pt,minimum height=6mm,fill=white,minimum width=12mm]
\tikzstyle{large map conj}=[draw,shape=NWbox,inner sep=2pt,minimum height=6mm,fill=white,minimum width=12mm]
\tikzstyle{very large map}=[draw,shape=NEbox,inner sep=2pt,minimum height=6mm,fill=white,minimum width=17mm]

\tikzstyle{medium dmap}=[draw,doubled,shape=NEbox,inner sep=2pt,minimum height=6mm,fill=white,minimum width=7mm]
\tikzstyle{medium dmap dag}=[draw,doubled,shape=SEbox,inner sep=2pt,minimum height=6mm,fill=white,minimum width=7mm]
\tikzstyle{medium dmap adj}=[draw,doubled,shape=SEbox,inner sep=2pt,minimum height=6mm,fill=white,minimum width=7mm]
\tikzstyle{medium dmap trans}=[draw,doubled,shape=SWbox,inner sep=2pt,minimum height=6mm,fill=white,minimum width=7mm]
\tikzstyle{medium dmap conj}=[draw,doubled,shape=NWbox,inner sep=2pt,minimum height=6mm,fill=white,minimum width=7mm]
\tikzstyle{semilarge dmap}=[draw,doubled,shape=NEbox,inner sep=2pt,minimum height=6mm,fill=white,minimum width=9.5mm]
\tikzstyle{semilarge dmap trans}=[draw,doubled,shape=SWbox,inner sep=2pt,minimum height=6mm,fill=white,minimum width=9.5mm]
\tikzstyle{semilarge dmap adj}=[draw,doubled,shape=SEbox,inner sep=2pt,minimum height=6mm,fill=white,minimum width=9.5mm]
\tikzstyle{semilarge dmap dag}=[draw,doubled,shape=SEbox,inner sep=2pt,minimum height=6mm,fill=white,minimum width=9.5mm]
\tikzstyle{semilarge dmap conj}=[draw,doubled,shape=NWbox,inner sep=2pt,minimum height=6mm,fill=white,minimum width=9.5mm]
\tikzstyle{large dmap}=[draw,doubled,shape=NEbox,inner sep=2pt,minimum height=6mm,fill=white,minimum width=12mm]
\tikzstyle{large dmap conj}=[draw,doubled,shape=NWbox,inner sep=2pt,minimum height=6mm,fill=white,minimum width=12mm]
\tikzstyle{large dmap trans}=[draw,doubled,shape=SWbox,inner sep=2pt,minimum height=6mm,fill=white,minimum width=12mm]
\tikzstyle{large dmap adj}=[draw,doubled,shape=SEbox,inner sep=2pt,minimum height=6mm,fill=white,minimum width=12mm]
\tikzstyle{large dmap dag}=[draw,doubled,shape=SEbox,inner sep=2pt,minimum height=6mm,fill=white,minimum width=12mm]
\tikzstyle{very large dmap}=[draw,doubled,shape=NEbox,inner sep=2pt,minimum height=6mm,fill=white,minimum width=19.5mm]

\tikzstyle{muxbox}=[draw,shape=rectangle,minimum height=3mm,minimum width=3mm,fill=white]
\tikzstyle{dmuxbox}=[muxbox,doubled]

\tikzstyle{box}=[draw,shape=rectangle,inner sep=2pt,minimum height=6mm,minimum width=6mm,fill=white]
\tikzstyle{dbox}=[draw,doubled,shape=rectangle,inner sep=2pt,minimum height=6mm,minimum width=6mm,fill=white]
\tikzstyle{dmap}=[draw,doubled,shape=NEbox,inner sep=2pt,minimum height=6mm,fill=white]
\tikzstyle{dmapdag}=[draw,doubled,shape=SEbox,inner sep=2pt,minimum height=6mm,fill=white]
\tikzstyle{dmapadj}=[draw,doubled,shape=SEbox,inner sep=2pt,minimum height=6mm,fill=white]
\tikzstyle{dmaptrans}=[draw,doubled,shape=SWbox,inner sep=2pt,minimum height=6mm,fill=white]
\tikzstyle{dmapconj}=[draw,doubled,shape=NWbox,inner sep=2pt,minimum height=6mm,fill=white]

\tikzstyle{ddmap}=[draw,doubled,dashed,shape=NEbox,inner sep=2pt,minimum height=6mm,fill=white]
\tikzstyle{ddmapdag}=[draw,doubled,dashed,shape=SEbox,inner sep=2pt,minimum height=6mm,fill=white]
\tikzstyle{ddmapadj}=[draw,doubled,dashed,shape=SEbox,inner sep=2pt,minimum height=6mm,fill=white]
\tikzstyle{ddmaptrans}=[draw,doubled,dashed,shape=SWbox,inner sep=2pt,minimum height=6mm,fill=white]
\tikzstyle{ddmapconj}=[draw,doubled,dashed,shape=NWbox,inner sep=2pt,minimum height=6mm,fill=white]

\boxshape{sNEbox}{0pt}{3pt}
\boxshape{sSEbox}{0pt}{-3pt}
\boxshape{sNWbox}{3pt}{0pt}
\boxshape{sSWbox}{-3pt}{0pt}
\tikzstyle{smap}=[draw,shape=sNEbox,fill=white]
\tikzstyle{smapdag}=[draw,shape=sSEbox,fill=white]
\tikzstyle{smapadj}=[draw,shape=sSEbox,fill=white]
\tikzstyle{smaptrans}=[draw,shape=sSWbox,fill=white]
\tikzstyle{smapconj}=[draw,shape=sNWbox,fill=white]

\tikzstyle{dsmap}=[draw,dashed,shape=sNEbox,fill=white]
\tikzstyle{dsmapdag}=[draw,dashed,shape=sSEbox,fill=white]
\tikzstyle{dsmaptrans}=[draw,dashed,shape=sSWbox,fill=white]
\tikzstyle{dsmapconj}=[draw,dashed,shape=sNWbox,fill=white]

\boxshape{mNEbox}{0pt}{10pt}
\boxshape{mSEbox}{0pt}{-10pt}
\boxshape{mNWbox}{10pt}{0pt}
\boxshape{mSWbox}{-10pt}{0pt}
\tikzstyle{mmap}=[draw,shape=mNEbox]
\tikzstyle{mmapdag}=[draw,shape=mSEbox]
\tikzstyle{mmaptrans}=[draw,shape=mSWbox]
\tikzstyle{mmapconj}=[draw,shape=mNWbox]

\tikzstyle{mmapgray}=[draw,fill=gray!40!white,shape=mNEbox]
\tikzstyle{smapgray}=[draw,fill=gray!40!white,shape=sNEbox]

\makeatletter
\pgfdeclareshape{cornerpoint}{
\inheritsavedanchors[from=rectangle] % this is nearly a rectangle
\inheritanchorborder[from=rectangle]
\inheritanchor[from=rectangle]{center}
\inheritanchor[from=rectangle]{north}
\inheritanchor[from=rectangle]{south}
\inheritanchor[from=rectangle]{west}
\inheritanchor[from=rectangle]{east}
\backgroundpath{% this is new
\southwest \pgf@xa=\pgf@x \pgf@ya=\pgf@y
\northeast \pgf@xb=\pgf@x \pgf@yb=\pgf@y

\pgfmathsetmacro{\pgf@shorten@left}{\pgfkeysvalueof{/tikz/shorten left}}
\pgfmathsetmacro{\pgf@shorten@right}{\pgfkeysvalueof{/tikz/shorten right}}

\pgfpathmoveto{\pgfpoint{0.5 * (\pgf@xa + \pgf@xb)}{\pgf@ya - 5pt}}
\pgfpathlineto{\pgfpoint{\pgf@xa - 8pt + \pgf@shorten@left}{\pgf@yb - 1.5 * \pgf@shorten@left}}
\pgfpathlineto{\pgfpoint{\pgf@xa - 8pt + \pgf@shorten@left}{\pgf@yb}}
\pgfpathlineto{\pgfpoint{\pgf@xb + 8pt - \pgf@shorten@right}{\pgf@yb}}
\pgfpathlineto{\pgfpoint{\pgf@xb + 8pt - \pgf@shorten@right}{\pgf@yb - 1.5 * \pgf@shorten@right}}
\pgfpathclose
}
}

\pgfdeclareshape{cornercopoint}{
\inheritsavedanchors[from=rectangle] % this is nearly a rectangle
\inheritanchorborder[from=rectangle]
\inheritanchor[from=rectangle]{center}
\inheritanchor[from=rectangle]{north}
\inheritanchor[from=rectangle]{south}
\inheritanchor[from=rectangle]{west}
\inheritanchor[from=rectangle]{east}
\backgroundpath{% this is new
\southwest \pgf@xa=\pgf@x \pgf@ya=\pgf@y
\northeast \pgf@xb=\pgf@x \pgf@yb=\pgf@y

\pgfmathsetmacro{\pgf@shorten@left}{\pgfkeysvalueof{/tikz/shorten left}}
\pgfmathsetmacro{\pgf@shorten@right}{\pgfkeysvalueof{/tikz/shorten right}}

\pgfpathmoveto{\pgfpoint{0.5 * (\pgf@xa + \pgf@xb)}{\pgf@yb + 5pt}}
\pgfpathlineto{\pgfpoint{\pgf@xa - 8pt + \pgf@shorten@left}{\pgf@ya + 1.5 * \pgf@shorten@left}}
\pgfpathlineto{\pgfpoint{\pgf@xa - 8pt + \pgf@shorten@left}{\pgf@ya}}
\pgfpathlineto{\pgfpoint{\pgf@xb + 8pt - \pgf@shorten@right}{\pgf@ya}}
\pgfpathlineto{\pgfpoint{\pgf@xb + 8pt - \pgf@shorten@right}{\pgf@ya + 1.5 * \pgf@shorten@right}}
\pgfpathclose
}
}

\makeatother

\pgfkeyssetvalue{/tikz/shorten left}{0pt}
\pgfkeyssetvalue{/tikz/shorten right}{0pt}

\tikzstyle{kpoint common}=[draw,fill=white,inner sep=1pt,minimum height=4mm]
\tikzstyle{kpoint}=[shape=cornerpoint,shorten left=5pt,kpoint common]
\tikzstyle{kpoint adjoint}=[shape=cornercopoint,shorten left=5pt,kpoint common]
\tikzstyle{kpoint conjugate}=[shape=cornerpoint,shorten right=5pt,kpoint common]
\tikzstyle{kpoint transpose}=[shape=cornercopoint,shorten right=5pt,kpoint common]
\tikzstyle{kpoint symm}=[shape=cornerpoint,shorten left=5pt,shorten right=5pt,kpoint common]

\tikzstyle{black kpoint}=[shape=cornerpoint,shorten left=5pt,kpoint common,fill=black,font=\color{white}]
\tikzstyle{black kpoint adjoint}=[shape=cornercopoint,shorten left=5pt,kpoint common,fill=black,font=\color{white}]
\tikzstyle{black kpointadj}=[shape=cornercopoint,shorten left=5pt,kpoint common,fill=black,font=\color{white}]

\tikzstyle{black dkpoint}=[shape=cornerpoint,shorten left=5pt,kpoint common,fill=black, doubled,font=\color{white}]
\tikzstyle{black dkpoint adjoint}=[shape=cornercopoint,shorten left=5pt,kpoint common,fill=black, doubled,font=\color{white}]
\tikzstyle{black dkpointadj}=[shape=cornercopoint,shorten left=5pt,kpoint common,fill=black, doubled,font=\color{white}]

\tikzstyle{kpointdag}=[kpoint adjoint]
\tikzstyle{kpointadj}=[kpoint adjoint]
\tikzstyle{kpointconj}=[kpoint conjugate]
\tikzstyle{kpointtrans}=[kpoint transpose]

\tikzstyle{big kpoint}=[kpoint, minimum width=1.2 cm, minimum height=8mm, inner sep=4pt, text depth=3mm]

\tikzstyle{wide kpoint}=[kpoint, minimum width=1 cm, inner sep=2pt]%, text depth=-0.7 mm]
\tikzstyle{wide kpointdag}=[kpointdag, minimum width=1 cm, inner sep=2pt]%, text depth=0.7 mm]
\tikzstyle{wide kpointconj}=[kpointconj, minimum width=1 cm, inner sep=2pt]%, text depth=-0.7 mm]
\tikzstyle{wide kpointtrans}=[kpointtrans, minimum width=1 cm, inner sep=2pt]%, text depth=0.7 mm]

\tikzstyle{gray kpoint}=[kpoint,fill=gray!50!white]
\tikzstyle{gray kpointdag}=[kpointdag,fill=gray!50!white]
\tikzstyle{gray kpointadj}=[kpointadj,fill=gray!50!white]
\tikzstyle{gray kpointconj}=[kpointconj,fill=gray!50!white]
\tikzstyle{gray kpointtrans}=[kpointtrans,fill=gray!50!white]

\tikzstyle{gray dkpoint}=[kpoint,fill=gray!50!white,doubled]
\tikzstyle{gray dkpointdag}=[kpointdag,fill=gray!50!white,doubled]
\tikzstyle{gray dkpointadj}=[kpointadj,fill=gray!50!white,doubled]
\tikzstyle{gray dkpointconj}=[kpointconj,fill=gray!50!white,doubled]
\tikzstyle{gray dkpointtrans}=[kpointtrans,fill=gray!50!white,doubled]

\tikzstyle{white label}=[draw,fill=white,rectangle,inner sep=0.7 mm]
\tikzstyle{gray label}=[draw,fill=gray!50!white,rectangle,inner sep=0.7 mm]
\tikzstyle{black label}=[draw,fill=black,rectangle,inner sep=0.7 mm]

\tikzstyle{dkpoint}=[kpoint,doubled]
\tikzstyle{wide dkpoint}=[wide kpoint,doubled]
\tikzstyle{dkpointdag}=[kpoint adjoint,doubled]
\tikzstyle{wide dkpointdag}=[wide kpointdag,doubled]
\tikzstyle{dkcopoint}=[kpoint adjoint,doubled]
\tikzstyle{dkpointadj}=[kpoint adjoint,doubled]
\tikzstyle{dkpointconj}=[kpoint conjugate,doubled]
\tikzstyle{dkpointtrans}=[kpoint transpose,doubled]

\tikzstyle{kscalar}=[kpoint common, shape=EBox, inner xsep=-1pt, inner ysep=3pt,font=\small]
\tikzstyle{kscalarconj}=[kpoint common, shape=WBox, inner xsep=-1pt, inner ysep=3pt,font=\small]

 \tikzstyle{upground}=[circuit ee IEC,thick,ground,rotate=90,scale=2.5]
 \tikzstyle{downground}=[circuit ee IEC,thick,ground,rotate=-90,scale=2.5]
 \tikzstyle{bigground}=[regular polygon,regular polygon sides=3,draw=gray,scale=0.50,inner sep=-0.5pt,minimum width=10mm,fill=gray]
\tikzstyle{arrs}=[-latex,font=\small,auto]
\tikzstyle{arrow plain}=[arrs]
\tikzstyle{arrow dashed}=[dashed,arrs]
\tikzstyle{arrow bold}=[very thick,arrs]
\tikzstyle{arrow hide}=[draw=white!0,-]
\tikzstyle{arrow reverse}=[latex-]
\tikzstyle{cdnode}=[]

\definecolor{hexcolor0xa9a9a9}{rgb}{0.663,0.663,0.663}
\tikzstyle{GrayLine}=[dashed,draw=hexcolor0xa9a9a9]
\tikzstyle{gray}=[dashed,draw=hexcolor0xa9a9a9]

\theoremstyle{definition}
\newtheorem{theorem}{Theorem}[section]
\newtheorem*{theorem*}{Theorem}
\newtheorem{corollary}[theorem]{Corollary}

\newtheorem{prop}[theorem]{Proposition}

\newtheorem{defn}[theorem]{Definition}

\newtheorem{example*}[theorem]{Example*}
\newtheorem{examples*}[theorem]{Examples*}

\newtheorem{remark*}[theorem]{Remark*}

\DeclareMathOperator{\Tr}{Tr}

\begin{document}
\maketitle

\begin{abstract}
Density operators are one of the key ingredients of quantum theory. They can be constructed in two ways: via a convex sum of `doubled kets' (i.e.~mixing), and by tracing out part of a  `doubled' two-system ket (i.e.~dilation). Both constructions can be iterated, yielding new mathematical species that have already found applications outside physics. However, as we show in this paper, the iterated constructions no longer yield the same mathematical species. Hence, the constructions `mixing' and `dilation' themselves are by no means equivalent. Concretely, when applying the Choi-Jamiolkowski isomorphism to the second iteration, dilation produces arbitrary symmetric bipartite states, while mixing only yields the disentangled ones. All results are proven using diagrams, and hence they hold not only for quantum theory, but also for a much more general class of process theories.  This paper is the shorter version of the ArXiv paper \cite{Zwart2017}, all missing proofs can be found in the full version.
\end{abstract}

\section{Introduction}

In 1932, von Neumann introduced special operators, now called density operators, to describe statistical \em mixtures \em of quantum states \cite{VonNeumann1955}. Unlike classical probability distributions, density operators are able to describe mixtures that involve a superposition of states, making them suitable for quantum theory. Von Neumann noticed that these density operators also arise when part of a state describing a composite system is discarded, a.k.a.~\em dilation\em. So two conceptually different physical processes happen to yield the same mathematical species in the quantum formalism. In other words, density operators are two-faced.

Mathematically, the fact that these two faces are distinct shows in the corresponding constructions. Following \cite{Ashoush}, density operators representing statistical mixtures are constructed by first matching each vector (ket) in the mixture with its corresponding functional (bra), turning the vectors into operators. We call this \em doubling\em. Then, these operators are combined in a convex sum, forming the density operator:
\begin{align*}
  \left\{|\phi_i\rangle\right\}_i & \mapsto \left\{|\phi_i\rangle\langle\phi_i|\right\}_i \\
   & \mapsto
   \sum_i \ p_i\  |\phi_i\rangle\langle\phi_i|
\end{align*}
where all $p_i\geq 0$ and  $\sum_i p_i =1$.

Density operators originating from dilation are also constructed by doubling a vector: this time a vector in a space of form $A \otimes B$. Then, part of the resulting operator is traced out, yielding again a density operator:
\begin{align*}
  |\psi_{AB}\rangle & \mapsto  |\psi_{AB}\rangle\langle\psi_{AB}| \\
 & \mapsto Tr_B \left( |\psi_{AB}\rangle\langle\psi_{AB}| \right),
\end{align*}
where $|\psi_{AB}\rangle$ is a vector in Hilbert space $A \otimes B$. As both constructions yield density operators, one is tempted to think of these as equivalent, which indeed many physicists do.

Iterating these constructions yields new mathematical species that were called \emph{dual density operators} in \cite{Ashoush2015, Ashoush}. However, it turns out that the iterated versions of the constructions are no longer equivalent. In fact, as we will prove in sections \ref{section:notthesame} and \ref{section:subspace}, the dual density operators resulting from double mixing form a proper subspace of those resulting from double dilation. Note that we use the term `double' mixing/dilation to emphasise that this procedure involves another round of doubling, hence distinguishing from either (i) mixing further already mixed states, and (ii) discarding a second system after discarding a first, both of which of course still yield ordinary density operators.

Of course, the constructions can be iterated once more, and even more after that. Each iteration yields new mathematical species, and from the second iteration onwards, mixtures always form a proper
subspace of what is obtained by dilation. We illustrate the results of iterating both constructions any finite number of times in section \ref{section:generalisation}, expanding the work in \cite{Ashoush}, which already considered the general case for dilation but not for mixing.

A closer examination of the differences between the results of double mixing and double dilation reveals that double mixing always results in symmetric disentangled states, whereas doubly dilated states are highly symmetrical, but not necessarily disentangled. This difference hints at a possible classification of states that are either doubly mixed or doubly dilated, which would contribute to the characterisation of entangled states started by Horodecki\cite{Horodecki2007}. We make a start of this enterprise in sections \ref{section:characterisation_twicemix} and \ref{section:characterisation_twicedilation}.

In quantum theory, density operators provide enough structure to describe the currently known phenomena, so for physics there seems to be no direct use for the iterated constructions. The only  notable exception known to us is the study of the space of iterated dilated states as a generalised probabilistic theory  by Barnum and Barrett \cite{Barnum-Barrett}. However, recent developments in \em natural language processing \em (NLP) have found an interesting application for these generalised variants of density operators \cite{Ashoush2015, Ashoush}. Density operators first appeared in the NLP literature in \cite{blacoe2013quantum}, where they are mainly used to enlarge the parameter space. A conceptual grounding matching that of quantum theory is given in the work of Piedeleu, Kartsaklis et al.~\cite{Piedeleu2015}, where they are used as a model for ambiguous words, representing these words as a statistical mixture over their possible `pure'  meanings. Here, the overall setting was that of categorical compositional distributional (DisCoCat) models of meaning of Coecke, Sadrzadeh and Clark \cite{Coeckea}, which was itself also strongly inspired by quantum theory \cite{teleling}.

A second application of density operators in NLP is found in the work of Balkir, Bankova et al.~\cite{Balkir2016, bankova2016graded, Balkr}, where lexical entailment is modelled by exploiting the fact that density operators can be partially ordered \cite{CoeckeMartin, Weteringen}. Naturally, this brought the need for a model that could accommodate ambiguity and lexical entailment simultaneously. It was to this end that Ashoush and Coecke started iterating the constructions of density operators. The results of this paper will hopefully contribute to further developing such models for NLP.

The structure of this paper is as follows. After a brief explanation of the graphical notation used in this paper, we recap the two constructions of density operators as described in \cite{Ashoush2015}, using both traditional Dirac bra-ket notation and diagrams. Then, in section \ref{section:notthesame}, we show that the results of iterating these constructions twice are no longer identical. The results are however strongly related: one being a subspace of the other, which we prove in section \ref{section:subspace}. Next, we characterise what results from double mixing as a special class of disentangled states in section \ref{section:characterisation_twicemix}, and reveal the symmetries caused by double dilation in section \ref{section:characterisation_twicedilation}. Finally, in section \ref{section:generalisation} we analyse the general case of applying both constructions any finite number of times, which emphasises the difference between them.

\subsection{Graphical notation}\label{section:diagram_notation}

Our proofs use a diagrammatic language designed for categorical quantum mechanics \cite{Kindergarten, Selinger2007, CPV, CPaqPav}, building further on Penrose's notation \cite{Penrose1971}.  We use this graphical notation because it greatly simplifies otherwise tedious proofs, abstracting away from unimportant details. It also has the advantage that the results are true in a more general setting than just Hilbert spaces. For the reader unfamiliar with this graphical notation we have included a brief summary of the main components that feature in this paper in the ArXiv version of this paper \cite{Zwart2017}. For an extensive introduction we refer to the textbook \cite{CoeckeBOOK} or the shorter  paper version \cite{Coecke2015, Coecke2016} which is also self-contained.

\section{Double mixing and double dilation}\label{section:constuction}

We recap both mixing and dilation, and show how these constructions can be iterated.

\subsection{Mixing}\label{section:mixing}
Given a set of $n$ normalised vectors $|\phi_i\rangle$ in a finite-dimensional Hilbert space $A$, and a probability distribution $\{p_i\}^{i\leq n}$, we form the density operator representing the mixture of these vectors as follows:
\begin{align*}
  (\{ |\phi_i\rangle\},\{ p_i\}) \mapsto \rho = & \sum_{i=1}^{n} p_i |\phi_i\rangle \langle \phi_i|
\end{align*}
In the diagrammatic language:
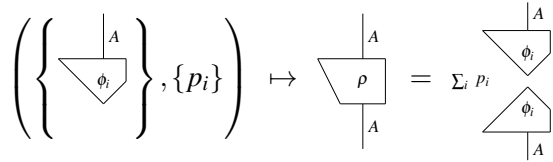
\[
\left(\left\{ \scalebox{0.6}{\input{phi_i_state.tikz}} \right\},\left\{ p_i \right\}\right)\;\; \mapsto \;\;\scalebox{0.6}{\input{rho_process.tikz}} \;\; = \;\; \scalebox{0.6}{\input{resultofmixing_process.tikz}}
\]
Alternatively, we could express $\rho$ as a vector in $A\otimes\overline{A}$. The benefit of obtaining a vector rather than an operator is that we get a construction that can be iterated, since it sends vectors to vectors:
\begin{align}\label{construction:sum}
(\{ |\phi_i\rangle\},\{ p_i\}) \mapsto ||\rho\rangle\rangle = & \sum_{i=1}^{n} p_i |\phi_i\rangle\overline{|\phi_i\rangle}
\end{align}
Here $\overline{|\phi_i\rangle}$ is the conjugate of $|\phi_i\rangle$, $|\phi_i\rangle\overline{|\phi_i\rangle}$ is shorthand for $|\phi_i\rangle \otimes \overline{|\phi_i\rangle}$ and notation $||\cdot\rangle\rangle$ is used to remind us that $\rho$ is a vector in Hilbert space $\hat{A} = A\otimes \overline{A}$ instead of $A$.

Diagrammatically, this construction translates as:
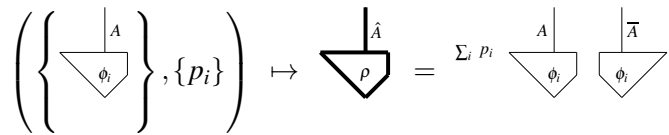
\[
\left(\left\{ \scalebox{0.6}{\input{phi_i_state.tikz}} \right\},\left\{ p_i \right\}\right)\;\; \mapsto \;\;\scalebox{0.6}{\input{rho_state.tikz}} \;\; = \;\; \scalebox{0.6}{\input{resultofmixing.tikz}}
\]

\newpage
A second iteration of (\ref{construction:sum}) with $m$ density vectors $||\rho_k\rangle\rangle = \sum_{i=1}^{n_k} p_{ki} |\phi_{ki}\rangle \overline{|\phi_{ki}\rangle}$ and a probability distribution $\{r_k\}$ yields:
\begin{align}
  |||\Psi\rangle\rangle\rangle = & \sum_{k=1}^{m} r_k ||\rho_k\rangle\rangle \overline{||\rho_k\rangle\rangle} \nonumber
 \\
  = & \sum_{k=1}^{m} r_k \left( \sum_{i=1}^{n_k} p_{ki} |\phi_{ki}\rangle \overline{|\phi_{ki}\rangle}\right) \left( \overline{\sum_{j=1}^{n_k} p_{kj} |\phi_{kj}\rangle \overline{|\phi_{kj}\rangle}}\right)\nonumber \\
  = & \sum_{k=1}^{m}\sum_{i=1}^{n_k}\sum_{j=1}^{n_k} r_k p_{ki} p_{kj} |\phi_{ki}\rangle \overline{|\phi_{ki}\rangle}|\phi_{kj}\rangle \overline{|\phi_{kj}\rangle}\label{result:doublemix}
\end{align}
Where $|||\Psi\rangle\rangle\rangle$ is a vector in $A\otimes\overline{A}\otimes A\otimes\overline{A} = \hat{\hat{A}}$. The accompanying diagram is (only showing the result, the dotted lines indicate the vectors $||\rho_k\rangle\rangle$ and $\overline{||\rho_k\rangle\rangle}$):
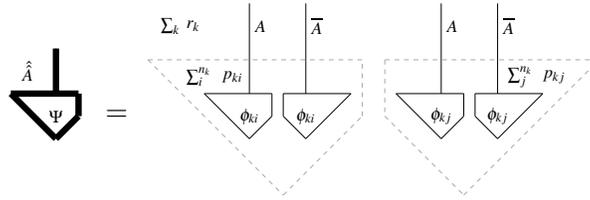
\[
\scalebox{0.6}{\input{Psi_state.tikz}} \;\; = \;\;\scalebox{0.6}{\input{resultofmixing_twice.tikz}}
\]

To make the diagram look prettier and to make it easier to compare to later results, we can hide the summations over $i$ and $j$ inside caps (wires $B$ in the diagram below) and summation over $k$ inside a four-legged spider (wire $C$)%, following the procedure `changing sums into spiders' explained in section \ref{section:changing_sum_into_spider}.
. The individual $\phi_{ki}$ will be no longer visible; they are absorbed into a general $\phi$:
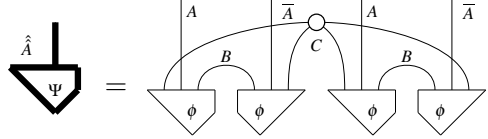
\[
\scalebox{0.6}{\input{Psi_state.tikz}} \;\; = \;\;\scalebox{0.6}{\input{doublemixstate_phi.tikz}}
\]
We call a vector resulting from twice applying construction (\ref{construction:sum}) \emph{doubly mixed}.

\subsection{Dilation}

On the other hand, if we have a vector $|\phi_{AB}\rangle$ in space $A\otimes B$, we can form the operator $|\phi_{AB}\rangle\langle\phi_{AB} |$ and trace out $B$:
\begin{align*}
 |\phi_{AB}\rangle \mapsto \rho' = & \Tr_B |\phi_{AB}\rangle\langle\phi_{AB}|
\end{align*}
Or as a diagram:
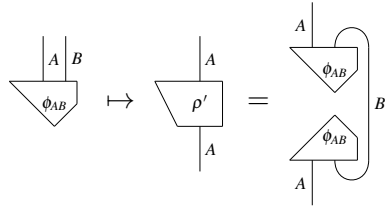
\[
\scalebox{0.6}{\input{phi_AB_state.tikz}} \;\; \mapsto\;\; \scalebox{0.6}{\input{rho_prime_process.tikz}}\;\; = \;\;\scalebox{0.6}{\input{Trace_B_Phi_AB.tikz}}
\]
If we would rather have a vector, this construction becomes, for some orthonormal basis $\{|e^B_i\rangle\}$ of $B$:
\begin{align}\label{construction:reduce}
 |\phi_{AB}\rangle \mapsto ||\rho'\rangle\rangle = &  \sum_{i=0}^{\dim(B)-1}\langle e^B_i |\phi_{AB}\rangle \overline{\langle e^B_i |\phi_{AB}\rangle}
\end{align}

The corresponding diagram is:
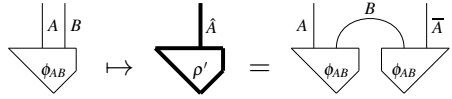
\[
\scalebox{0.6}{\input{phi_AB_state.tikz}} \;\; \mapsto\;\; \scalebox{0.6}{\input{rho_prime_state.tikz}}\;\; = \;\;\scalebox{0.6}{\input{reducedstate.tikz}}
\]

To iterate (\ref{construction:reduce}), we need the Hilbert space $A$ to be of form $A \otimes C$, so that after reducing a second time, the result is still a vector instead of a number. So suppose that our original vector was $|\phi_{ABC}\rangle$. Applying (\ref{construction:reduce}) using space $B$ then yields $||\rho'_{\hat{A}\hat{C}}\rangle\rangle$ in $A \otimes C \otimes \overline{C} \otimes \overline{A}$. Applying (\ref{construction:reduce}) again, now using space $\hat{C} = C \otimes \overline{C}$ results in:

\begin{align}
  |||\Psi'\rangle\rangle\rangle = & \sum_{k,l}\langle\langle e^{\hat{C}}_{k,l} ||\rho'_{\hat{A}\hat{C}}\rangle\rangle \overline{\langle\langle e^{\hat{C}}_{k,l} ||\rho'_{\hat{A}\hat{C}}\rangle\rangle} \nonumber \\
  = &  \sum_{k}\sum_{l}\left( \sum_{i}\langle e^C_k| \langle e^B_i|\phi_{ABC}\rangle \overline{\langle e^C_l| \langle e^B_i|\phi_{ABC}\rangle} \right) \left(\overline{\sum_{j}\langle e^C_k| \langle e^B_j|\phi_{ABC}\rangle \overline{\langle e^C_l| \langle e^B_j|\phi_{ABC}\rangle}} \right) \nonumber \\
  = & \sum_{k}\sum_{l}\sum_{i}\sum_{j}\langle e^C_k| \langle e^B_i|\phi_{ABC}\rangle \overline{\langle e^C_l| \langle e^B_i|\phi_{ABC}\rangle} \langle e^C_l| \langle e^B_j|\phi_{ABC}\rangle \overline{\langle e^C_k| \langle e^B_j|\phi_{ABC}\rangle} \label{result:doublereduce}
\end{align}
The corresponding diagram is:
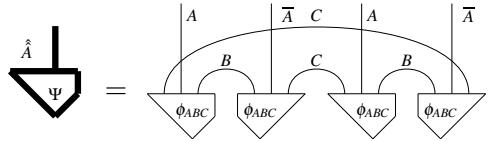
\[
\scalebox{0.6}{\input{Psi_state.tikz}} \;\; = \;\;\scalebox{0.6}{\input{dualdensitymatrix1_phi_ABC.tikz}}
\]
We call such vectors \emph{doubly dilated}.

\section{Counterexample to equivalence}\label{section:notthesame}

Comparing expressions (\ref{construction:sum}) and (\ref{result:doublemix}) to (\ref{construction:reduce}) and (\ref{result:doublereduce}), it is not at all obvious that mixing and dilation could be equivalent; indeed, they are not! Although it is possible to prove that both constructions give the same results when applied just a single time, this is no longer true when the constructions are iterated. As a counterexample, we give a vector resulting from double dilation that cannot be the result of double mixing.

\begin{theorem}\label{theorem_strictsubspace}
There exist vectors resulting from double dilation that cannot be written as vectors resulting from double mixing.
\end{theorem}

The following doubly dilated vector provides a counterexample. %The full proof can be found in \cite{Zwart2017}.
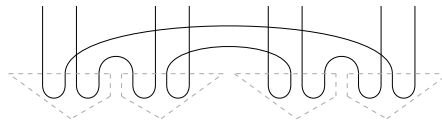
\[
  \scalebox{0.6}{\input{counterexample_state.tikz}}
\]

\begin{proof}
Consider the following diagram (left), which is Choi-Jamiolkowski isomorphic to the counterexample mentioned above. We will show that it cannot be formed with either of the two diagrams on the right, which are the only two possible operators resulting from applying the Choi-Jamiolkowski isomorphism to a vector resulting from double mixing (see also equations \ref{dm1} and \ref{dm2} in section \ref{section:characterisation_twicedilation} below).
\[
\scalebox{1.0}{\input{notadoublemixture1.tikz}} \;\;\;\notin\;\;\; \left\{ \scalebox{0.6}{\input{doublemixmorphism2.tikz}}\;\;,\;\; \scalebox{0.6}{\input{doublemixmorphism2otherway.tikz}}  \right\}
\]
First, we prove that the identity morphism cannot be written as a diagram of the form of the rightmost diagram above. To show this, suppose that the identity \emph{can} be written in that form:
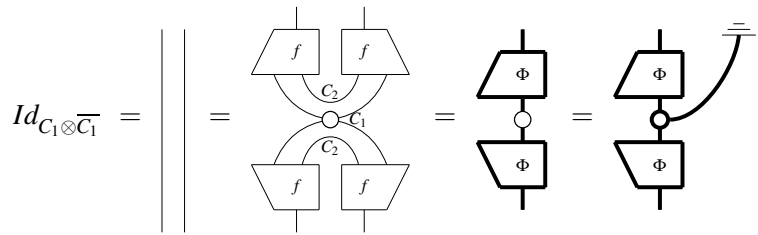
\begin{align*}
  Id_{C_1 \otimes \overline{C_1}}\;\; &= \;\; \scalebox{0.6}{\begin{tikzpicture}
	\begin{pgfonlayer}{nodelayer}
		\node [style=none] (0) at (-0.5000002, -5) {};
		\node [style=none] (1) at (-0.5, 4) {};
		\node [style=none] (2) at (0.5000002, 4) {};
		\node [style=none] (3) at (0.5000002, -5) {};
	\end{pgfonlayer}
	\begin{pgfonlayer}{edgelayer}
		\draw (2.center) to (3.center);
		\draw (1.center) to (0.center);
	\end{pgfonlayer}
\end{tikzpicture}} \;\; = \;\; \scalebox{0.6}{\input{doublemixmorphism2otherway.tikz}} \;\; = \;\;
  \scalebox{0.6}{\input{doublemixmorphism2otherway_dubbelgevouwen_spider.tikz}}\;\; = \;\; \scalebox{0.6}{\input{doublemixmorphism2otherway_dubbelgevouwen.tikz}}
\end{align*}
Here, the third equality is just rewriting the wires and boxes using thick lines\footnote{for an explanation about thick line notation and spiders, see section 2 of the ArXiv version of this paper \cite{Zwart2017}}. By doing this, the classical spider turns into a bastard spider. In the last diagram, the bastard spider used fission to turn into a quantum spider with a wire that is discarded.

We can now use the following theorem:
\begin{quotation}
``If a reduced operator $\Psi_{reduced}$ (an operator with one of its outputs discarded) is pure (can be written as a tensor product of some operator and its conjugate: $\Psi_{reduced} = \Phi \otimes \overline{\Phi}$), then the original operator can be written as a tensor product of that pure operator and a (possibly impure) vector: $\Psi = \Phi \otimes \overline{\Phi} \otimes ||\rho\rangle\rangle$'' \cite[Proposition 6.78]{CoeckeBOOK}.
\end{quotation}
Applying this theorem to the rightmost diagram above gives:
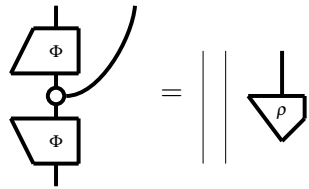
\[
\scalebox{0.6}{\input{doublemixmorphism2otherway_dubbelgevouwen_2.tikz}}\;\; = \;\; \scalebox{0.6}{\input{identity_and_state.tikz}}
\]
For some vector $\rho$. As $\rho$ is non-zero, there exists a basis vector $\langle e_i |$ such that the following is nonzero:
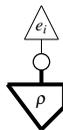
\[
\scalebox{0.6}{\input{rho_effect_i.tikz}}
\]
Combining this with the above:
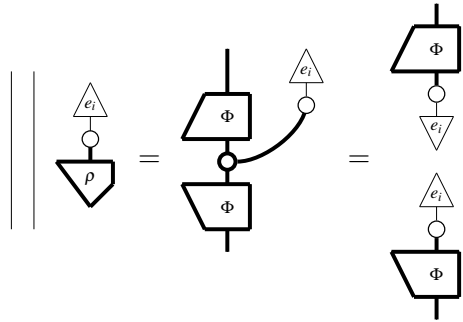
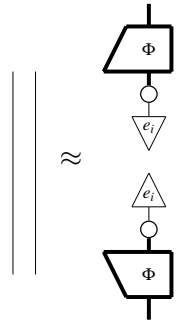
\begin{align*}
  \scalebox{0.6}{\input{identity_and_state_with_i.tikz}} \;\; & = \;\;\scalebox{0.6}{\input{doublemixmorphism2otherway_dubbelgevouwen_with_i.tikz}} \;\;  = \;\; \scalebox{0.6}{\input{doublemixmorphism2otherway_dubbelgevouwen_disconnect.tikz}}
\end{align*}
And hence:
\[
\scalebox{0.6}{\begin{tikzpicture}
	\begin{pgfonlayer}{nodelayer}
		\node [style=none] (0) at (-0.5000002, -5) {};
		\node [style=none] (1) at (-0.5, 4) {};
		\node [style=none] (2) at (0.5000002, 4) {};
		\node [style=none] (3) at (0.5000002, -5) {};
	\end{pgfonlayer}
	\begin{pgfonlayer}{edgelayer}
		\draw (2.center) to (3.center);
		\draw (1.center) to (0.center);
	\end{pgfonlayer}
\end{tikzpicture}} \;\; \approx  \;\; \scalebox{0.6}{\input{doublemixmorphism2otherway_dubbelgevouwen_disconnect.tikz}}
\]
In other words, the identity $\circ$-separates, which is non-sense \cite{CoeckeBOOK}.
Therefore, the identity cannot be of form:
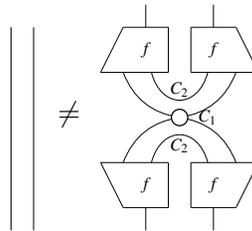
\[
\scalebox{0.6}{\begin{tikzpicture}
	\begin{pgfonlayer}{nodelayer}
		\node [style=none] (0) at (-0.5000002, -5) {};
		\node [style=none] (1) at (-0.5, 4) {};
		\node [style=none] (2) at (0.5000002, 4) {};
		\node [style=none] (3) at (0.5000002, -5) {};
	\end{pgfonlayer}
	\begin{pgfonlayer}{edgelayer}
		\draw (2.center) to (3.center);
		\draw (1.center) to (0.center);
	\end{pgfonlayer}
\end{tikzpicture}} \;\; \neq  \;\;\scalebox{0.6}{\input{doublemixmorphism2otherway.tikz}}
\]
With some wire-bending (using the Choi-Jamiolkowski isomorphism), we then see that the diagram composed of a cap followed by a cup cannot be of form:
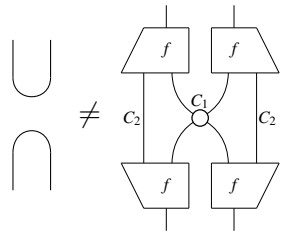
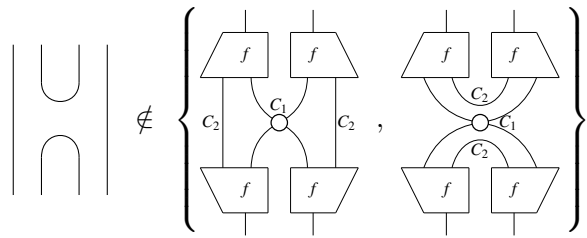
\[
\scalebox{1.0}{\input{capcomposedwithcup.tikz}}\;\; \neq \;\;\scalebox{0.6}{\input{doublemixmorphism2.tikz}}
\]
And so:
\[
\scalebox{1.0}{\input{notadoublemixture1.tikz}} \;\;\notin \;\;\left\{ \scalebox{0.6}{\input{doublemixmorphism2.tikz}}\;\;,\;\; \scalebox{0.6}{\input{doublemixmorphism2otherway.tikz}}  \right\}
\]
This proves that the doubly dilated vector shown at the beginning of this proof cannot be the result of double mixing. Therefore, it provides the counterexample we needed to show that double dilation and double mixing are indeed non-equivalent constructions.
\end{proof}

\section{Double mixing $\subseteq$ double dilation}\label{section:subspace}

In the previous section, we gave an example of a vector resulting from double dilation that could not result from double mixing. The converse, however, does hold: every vector resulting from double mixing can be obtained from double dilation. In other words, double mixing yields a proper subspace of double dilation. The proof is a beautiful example of diagrammatic reasoning with spiders. %; for those who prefer bra-ket notation instead, the appendix contains a proof sketch in which we manipulate the bra-ket expression (\ref{result:doublemix}) to match that of (\ref{result:doublereduce}).

\begin{theorem}\label{theorem_subspace}
Every vector resulting from double mixing also results from double dilation.
\end{theorem}

\begin{proof}
Given a vector resulting from double mixing, we can use the spider fission %(see section \ref{section:spiders})
to make four spiders:
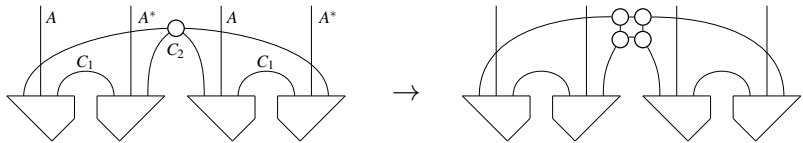
\[
\scalebox{0.6}{\input{doublemixstate.tikz}} \;\; \rightarrow \;\;
\scalebox{0.6}{\input{doublemixstate_2.tikz}}
\]
Moving these spiders closer to the four boxes gives a familiar picture. Absorbing the spiders into the boxes yields the vector resulting from double mixing as a vector resulting from double dilation:
\[
\rightarrow \;\; \scalebox{0.6}{\input{doublemixstate_3.tikz}} \;\; \rightarrow \scalebox{0.6}{\input{doublemixstate_4.tikz}}
\]
\end{proof}

By Theorems \ref{theorem_strictsubspace} and \ref{theorem_subspace} we then obtain:

\begin{corollary}
Double mixing yields a proper subspace of double dilation.
\end{corollary}

%\newpage
\section{Double mixing $\simeq$ disentangled states}\label{section:characterisation_twicemix}

\begin{defn}
Following \cite{CoeckeBOOK}, \em disentangled \em bipartite states are those states with diagrams of form:
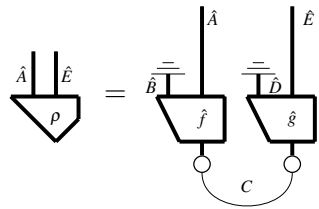
\[
\scalebox{0.6}{\input{rho_bipartite_state.tikz}} \;\; = \;\; \scalebox{0.6}{\input{general_disentangled_state.tikz}}
\]
\em Entangled \em states are those that are not disentangled.
\end{defn}
The intuition behind this idea is that disentangled states can share only classical information (the thin wire connecting the left and right halves of the diagram).

\begin{prop}
Vectors resulting from double mixing correspond to conjugate-symmetric
disentangled states.
\end{prop}
\begin{proof}
Consider a vector resulting from double mixing, and rewrite it using the thick-line notation: $||\hat{\phi}\rangle\rangle = |\phi\rangle\overline{|\phi\rangle}$:
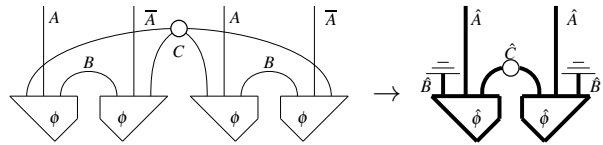
\[
\scalebox{0.6}{\input{doublemixstate_phi.tikz}} \;\; \rightarrow \;\; \scalebox{0.6}{\input{doublemixstate_phi_thick.tikz}}
\]
Next, move the spider downwards using the Choi-Jamiolkowski isomorphism, and use $\hat{f}$ for the operator resulting from the isomorphism applied to $||\hat{\phi}\rangle\rangle$. Then lastly, use spider fission to arrive at:
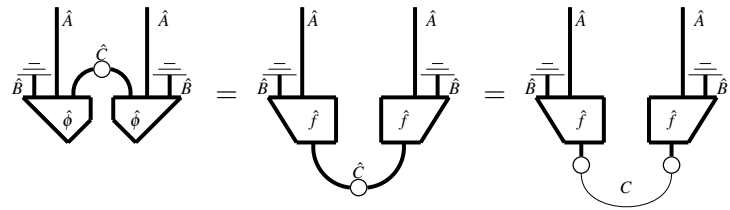
\[
\scalebox{0.6}{\input{doublemixstate_phi_thick.tikz}} \;\;=\;\; \scalebox{0.6}{\input{doublemixstate_phi_thick_down.tikz}} \;\;=\;\; \scalebox{0.6}{\input{doublemixstate_phi_thick_down_disentangled.tikz}}
\]
This is the general form of a disentangled state as described above, with only one restriction: the bipartite state has to be conjugate-symmetric.
\end{proof}

Contrast this with a vector resulting from double dilation:
Consider a vector resulting from double mixing and follow the same steps as above in rewriting:
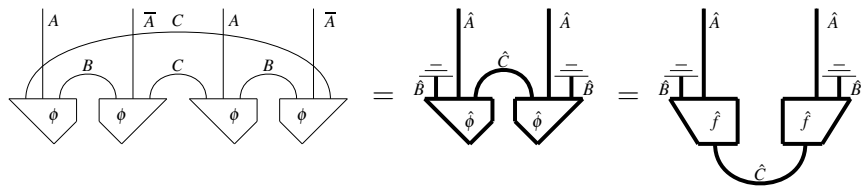
\[
\scalebox{0.6}{\input{dualdensitymatrix1_phi.tikz}} \;\; = \;\;
\scalebox{0.6}{\input{doubledilation_phi_thick.tikz}} \;\;=\;\; \scalebox{0.6} {\input{doubledilation_phi_thick_down.tikz}}
\]
This vector has the same symmetry as the one resulting from double mixing, but it is not necessarily disentangled. The symmetries introduced by mixing and dilation completely characterise the resulting vectors, which we show in the next section.

\section{The characterising symmetries of double dilation}\label{section:characterisation_twicedilation}

When we consider $|||\Psi\rangle\rangle\rangle$ from equation (\ref{result:doublemix}), there are two ways in which we can turn the vector $|||\Psi\rangle\rangle\rangle$ into an operator, both using the Choi-Jamiolkowski isomorphism:
\begin{align}
 |||\Psi\rangle\rangle\rangle = & \sum_{k=1}^{m}\sum_{i=1}^{n_k}\sum_{j=1}^{n_k} r_k p_{ki} p_{kj} |\phi_{ki}\rangle \overline{|\phi_{ki}\rangle}|\phi_{kj}\rangle \overline{|\phi_{kj}\rangle} \nonumber \\
  \text{operator}_1 = & \sum_{k=1}^{m}\sum_{i=1}^{n_k}\sum_{j=1}^{n_k} r_k p_{ki} p_{kj} |\phi_{ki}\rangle\overline{|\phi_{ki}\rangle} \langle\phi_{kj}| \overline{\langle\phi_{kj}|} \\
  \text{operator}_2 = & \sum_{k=1}^{m}\sum_{i=1}^{n_k}\sum_{j=1}^{n_k} r_k p_{ki} p_{kj} |\phi_{ki}\rangle|\phi_{kj}\rangle \langle\phi_{ki}| \langle\phi_{kj}|
\end{align}
Similarly for the doubly dilated vectors from equation (\ref{result:doublereduce}), for which we give the diagram expressions:
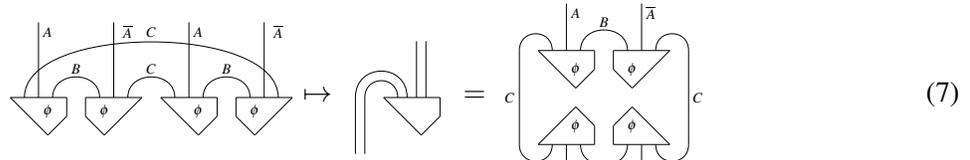
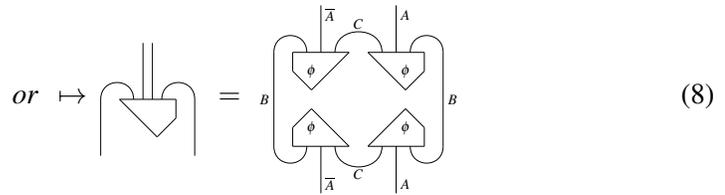
\begin{align}
  \scalebox{0.5}{\input{dualdensitymatrix1_phi.tikz}} \mapsto \;\; & \;\, \scalebox{0.5}{\input{densitymatrix1.tikz}} \;\; = \;\; \scalebox{0.5}{\input{dualdensitymatrix1_dm1.tikz}} \label{dm1} \\
   or\;\; \mapsto\;\; & \scalebox{0.5}{\input{densitymatrix2.tikz}} \;\; = \;\; \scalebox{0.5}{\input{dualdensitymatrix1_dm2.tikz}} \label{dm2}
\end{align}

In both cases, the resulting operator is positive semi-definite and self-adjoint. In other words, the operators are density operators. We call these \emph{the CJ-density operators of} the vector (from Choi-Jamiolkowski). The property of having two CJ-density operators completely characterises vectors resulting from double dilation.

\begin{theorem}\label{thm:characterisation2}
 Let $\phi$ be any normalised vector in any finite Hilbert space. Then $\phi$ has two CJ-density operators $CJ_1$ and $CJ_2$ iff $\phi$ is a result from double dilation.
\end{theorem}

The proof of this theorem can be found in the extended version of this paper on ArXiv \cite{Zwart2017}.

\section{Multiple iterations}\label{section:generalisation}

We generalise mixing and dilation by iterating both constructions not just twice, but any finite number of times. The iterated version of dilation has already been thoroughly studied in Ashoush's Master thesis \cite{Ashoush}. Here, we still give a sketch of the results of each iteration, so we can contrast them with the results from iterated mixing. For mixing, we also just give the resulting diagrams, trusting that the reader can imagine how to generalise the construction given in section \ref{section:mixing} to more than two iterations.

The $0^{th}$ iteration of both constructions is just doing nothing, so we have normal vectors:
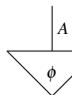
\[
\scalebox{0.6}{\input{phi_state.tikz}}
\]

\newpage
Mixing and dilation once turns these vectors into ones of form:
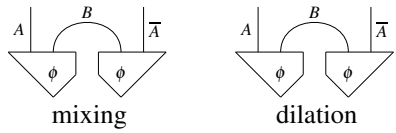
\begin{center}
  \begin{tabular}{ccc}
    \scalebox{0.6}{\input{reducedstate_nolabel.tikz}} &  & \scalebox{0.6}{\input{reducedstate_nolabel.tikz}} \\
    \small mixing &  & \small dilation \\
    \end{tabular}
\end{center}

The difference between the two shows in the second iteration:
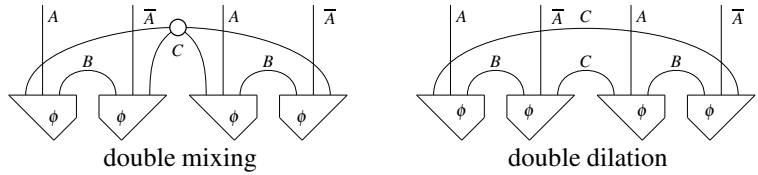
\begin{center}
  \begin{tabular}{ccc}
    \scalebox{0.6}{\input{doublemixstate_phi.tikz}} &  & \scalebox{0.6}{\input{dualdensitymatrix1_phi.tikz}} \\
    \small double mixing &  & \small double dilation \\
    \end{tabular}
\end{center}

We iterate both constructions a third time, first mixing:
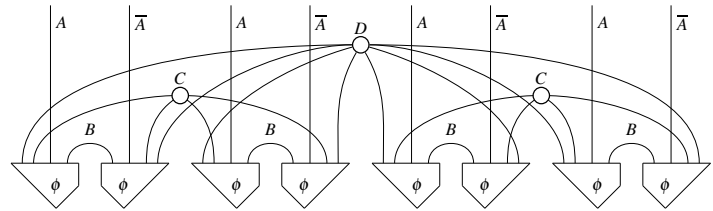
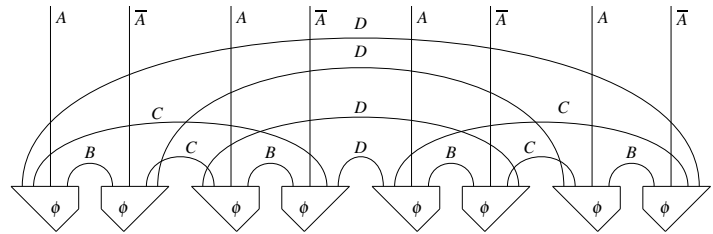
\[
\scalebox{0.6}{\input{3mixture_phi.tikz}}
\]

and then dilation:
\[
\scalebox{0.6}{\input{cpm3state_phi.tikz}}
\]

In general, the $n^{th}$ iteration of mixing takes the result from the previous iteration and tensors it with its conjugate. Then, all the resulting boxes are connected to a single new spider with $2^n$ legs.

On the other hand, the $n^{th}$ iteration of dilation does also take the tensor product of the previous iteration and its conjugate, but then makes $2^{n-1}$ nested connections (a rainbow), connecting each box from the last iteration to its counterpart in the conjugate half.
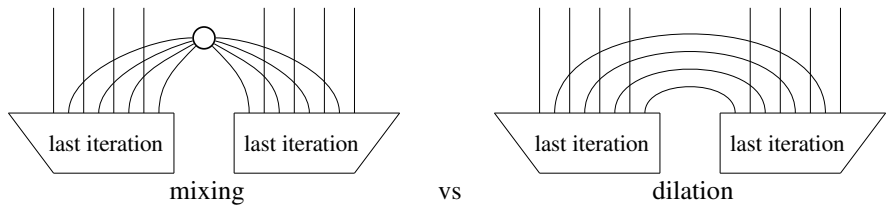
\begin{center}
  \begin{tabular}{ccc}
    \scalebox{0.8}{\input{general_mixing.tikz}} & & \scalebox{0.8}{\input{general_cpm.tikz}} \\
    \small mixing & \small vs & \small dilation \\
  \end{tabular}
\end{center}

\subsection{Always a strict subspace}

In every iteration except for the first, mixing yields a proper subspace of dilation. This emphasises again that mixing and dilation are two non-equivalent constructions.
\begin{theorem}
 For all $n \geq 2$, $n$ iterations of mixing yields a proper subspace of the result from $n$ iterations of dilation.
\end{theorem}

The proof is an easy generalisation of the proofs of Theorems \ref{theorem_subspace} and \ref{theorem_strictsubspace}.

\subsection{More symmetry}\label{section:general_symmetry}

Vectors resulting from double dilation were characterised by having two CJ-density operators. This suggests that those resulting from $n$ iterations of dilation are precisely those that have $n$ CJ-density operators, capturing the extra symmetry introduced by each iteration of dilation.

\begin{theorem}\label{thm:characterisation_n}
Let $\phi$ be any vector in a finite Hilbert space. Then $\phi$ has $n$ CJ-density operators
$CJ_1, \ldots, CJ_n$ iff $\phi$ is the result of $n$ iterations of dilation.
\end{theorem}

The proof is by induction. Theorem \ref{thm:characterisation2} provides the base case, the rest of the induction is included in the appendix of the extended version on ArXiv \cite{Zwart2017}. Note that the fact that $\phi$ has $n$ CJ-density operators immediately implies that its type is of form $(A\otimes\overline{A})^{2^{n-1}}$, that is, it is a vector in Hilbert space $(A\otimes\overline{A})^{2^{n-1}}$.
Of course, as mixing always yields a subspace of dilation, vectors resulting from $n$ iterations of mixing also have the extra symmetry properties. They stay disentangled in the way discussed in section \ref{section:characterisation_twicemix}.

\section{Discussion and outlook}

Although in the physics community it us usually assumed that dilation and mixing are one and the same thing, this is clearly not the case.  The heart of our result can simply be depicted as:
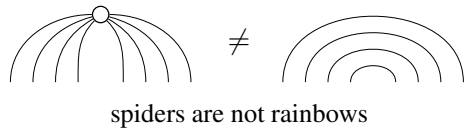
\begin{center}
  \begin{tabular}{ccc}
    \scalebox{0.6}{\input{general_spider_8legs.tikz}} & \raisebox{5mm}{$\neq$} & \scalebox{0.6} {\input{general_rainbow.tikz}} \\
      \multicolumn{3}{c}{\small spiders are not rainbows} \\
  \end{tabular}
\end{center}
That is, a convex sum over pure operators (mixing, yielding a spider diagram) is not the same as a partially traced out composite system (dilation, yielding a rainbow diagram), even though both constructions happen to coincide in the case of density operators (i.e.~the result of applying them only once), since:
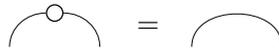
\begin{center}
  \begin{tabular}{ccc}
    \scalebox{0.6}{\begin{tikzpicture}
	\begin{pgfonlayer}{nodelayer}
		\node [style=none] (0) at (-2, 0) {};
		\node [style=none] (1) at (2, 0) {};
		\node [style=wn] (2) at (0, 1.5) {};
	\end{pgfonlayer}
	\begin{pgfonlayer}{edgelayer}
		\draw [in=180, out=90, looseness=1.00] (0.center) to (2);
		\draw [in=90, out=0, looseness=1.00] (2) to (1.center);
	\end{pgfonlayer}
\end{tikzpicture}} & \raisebox{2.5mm}{$=$} & \scalebox{0.6}{\begin{tikzpicture}
	\begin{pgfonlayer}{nodelayer}
		\node [style=none] (0) at (-2, 0) {};
		\node [style=none] (1) at (2, 0) {};
	\end{pgfonlayer}
	\begin{pgfonlayer}{edgelayer}
		\draw [bend left=90, looseness=1.25] (0.center) to (1.center);
	\end{pgfonlayer}
\end{tikzpicture}} \\
   %   \multicolumn{3}{c}{\small\hspace{-3cm}humanoid spiders are monochromatic rainbows\hspace{-3cm}} \\
  \end{tabular}
\end{center}

In physics, this result may impact axiomatic understanding of density matrices, and may also contribute to either crafting interesting toy theories, or adjoining extra variables to theories.

In NLP, it is worth considering which of the two, double mixing and double dilation, could serve as a model for both ambiguity and lexical entailment. Notice that in dictionaries, disambiguation of words is always first by hypernym, then by entailment. If this order is something that the model should reflect, then double mixing is a good choice: the asymmetry is reflected by the spiders appearing in the mixtures, causing a clear distinction between the first and second iterations of mixing. If however, this order of disambiguation in dictionaries is considered artificial, then the more general double dilation might be the preferred option.

The second result presented in this paper is the characterisation of both constructions. Dilation yields vectors that have $n$ CJ-density operators, which nicely exposes the symmetries introduced by the construction. Mixtures on the other hand, while having the same symmetries as dilated vectors, are special cases of disentangled states. This actually comes as no surprise: mixtures are almost by definition impure things. For future research, it would be ideal to find a characterisation for vectors resulting from double dilation that are \em not \em the result of double mixing.

As we mentioned in section \ref{section:diagram_notation}, the results in this paper apply in a more general setting than finite Hilbert spaces. To be precise, they hold in any \em spider category \em (dagger compact closed category with a Frobenious structure). One such category is the category of sets and relations (Rel). Oscar Cunningham and Dan Marsden have looked into the application of iterated dilation to the states in Rel\cite{Marsden15,Cunningham-privite}. In ongoing research, we are now applying iterated mixing to Rel as well. Hopefully, this will give us some hints about vectors resulting from double dilation but not from double mixing.

\bibliographystyle{eptcs}
\bibliography{lib-for-submission}{}

\begin{thebibliography}{10}
\providecommand{\bibitemdeclare}[2]{}
\providecommand{\surnamestart}{}
\providecommand{\surnameend}{}
\providecommand{\urlprefix}{Available at }
\providecommand{\url}[1]{\texttt{#1}}
\providecommand{\href}[2]{\texttt{#2}}
\providecommand{\urlalt}[2]{\href{#1}{#2}}
\providecommand{\doi}[1]{doi:\urlalt{http://dx.doi.org/#1}{#1}}
\providecommand{\bibinfo}[2]{#2}

\bibitemdeclare{phdthesis}{Ashoush2015}
\bibitem{Ashoush2015}
\bibinfo{author}{Daniela \surnamestart Ashoush\surnameend}
  (\bibinfo{year}{2015}): \emph{\bibinfo{title}{{Categorical Models of Meaning
  : Accommodating for Lexical Ambiguity and Entailment}}}.
\newblock \bibinfo{type}{Master's thesis}, \bibinfo{school}{Oxford University}.

\bibitemdeclare{article}{Ashoush}
\bibitem{Ashoush}
\bibinfo{author}{Daniela \surnamestart Ashoush\surnameend} \&
  \bibinfo{author}{Bob \surnamestart Coecke\surnameend} (\bibinfo{year}{2016}):
  \emph{\bibinfo{title}{{Dual Density Operators and Natural Language
  Meaning}}}.
\newblock {\sl \bibinfo{journal}{Electronic Proceedings in Theoretical Computer
  Science}} \bibinfo{volume}{221}, pp. \bibinfo{pages}{1--10},
  \doi{10.4204/EPTCS.221.1}.

\bibitemdeclare{article}{Balkr}
\bibitem{Balkr}
\bibinfo{author}{Esma \surnamestart Balkir\surnameend},
  \bibinfo{author}{Dimitri \surnamestart Kartsaklis\surnameend} \&
  \bibinfo{author}{Mehrnoosh \surnamestart Sadrzadeh\surnameend}:
  \emph{\bibinfo{title}{{Sentence Entailment in Compositional Distributional
  Semantics}}}.
\newblock {\sl \bibinfo{journal}{International Symposium on Artificial
  Intelligence and Mathematics (ISAIM)}}, \doi{10.1007/978-3-319-28678-5\_1}.
\newblock \urlprefix\url{https://arxiv.org/abs/1512.04419}.

\bibitemdeclare{inproceedings}{Balkir2016}
\bibitem{Balkir2016}
\bibinfo{author}{Esma \surnamestart Balkir\surnameend},
  \bibinfo{author}{Mehrnoosh \surnamestart Sadrzadeh\surnameend} \&
  \bibinfo{author}{Bob \surnamestart Coecke\surnameend} (\bibinfo{year}{2016}):
  \emph{\bibinfo{title}{{Distributional sentence entailment using density
  matrices}}}.
\newblock In: {\sl \bibinfo{booktitle}{Lecture Notes in Computer Science
  (including subseries Lecture Notes in Artificial Intelligence and Lecture
  Notes in Bioinformatics)}}, \bibinfo{volume}{9541}, pp.
  \bibinfo{pages}{1--22}, \doi{10.1007/978-3-319-28678-5\_1}.
\newblock \urlprefix\url{https://arxiv.org/abs/1506.06534}.

\bibitemdeclare{article}{bankova2016graded}
\bibitem{bankova2016graded}
\bibinfo{author}{Desislava \surnamestart Bankova\surnameend},
  \bibinfo{author}{Bob \surnamestart Coecke\surnameend},
  \bibinfo{author}{Martha \surnamestart Lewis\surnameend} \&
  \bibinfo{author}{Daniel \surnamestart Marsden\surnameend}
  (\bibinfo{year}{2016}): \emph{\bibinfo{title}{Graded Entailment for
  Compositional Distributional Semantics}}.
\newblock {\sl \bibinfo{journal}{CoRR}} \bibinfo{volume}{abs/1601.04908}.
\newblock \urlprefix\url{http://arxiv.org/abs/1601.04908}.

\bibitemdeclare{article}{Barnum-Barrett}
\bibitem{Barnum-Barrett}
\bibinfo{author}{Howard \surnamestart Barnum\surnameend} \&
  \bibinfo{author}{Jonathan \surnamestart Barrett\surnameend}:
  \bibinfo{note}{Private communication}.

\bibitemdeclare{inproceedings}{blacoe2013quantum}
\bibitem{blacoe2013quantum}
\bibinfo{author}{William \surnamestart Blacoe\surnameend},
  \bibinfo{author}{Elham \surnamestart Kashefi\surnameend} \&
  \bibinfo{author}{Mirella \surnamestart Lapata\surnameend}
  (\bibinfo{year}{2013}): \emph{\bibinfo{title}{{A Quantum-Theoretic Approach
  to Distributional Semantics}}}.
\newblock In: {\sl \bibinfo{booktitle}{North American Chapter of the
  Association for Computational Linguistics: Human Language Technologies (NAACL
  HLT)}}, pp. \bibinfo{pages}{847--857}.

\bibitemdeclare{article}{teleling}
\bibitem{teleling}
\bibinfo{author}{Stephen \surnamestart Clark\surnameend}, \bibinfo{author}{Bob
  \surnamestart Coecke\surnameend}, \bibinfo{author}{Edward \surnamestart
  Grefenstette\surnameend}, \bibinfo{author}{Stephen \surnamestart
  Pulman\surnameend} \& \bibinfo{author}{Mehrnoosh \surnamestart
  Sadrzadeh\surnameend} (\bibinfo{year}{2014}): \emph{\bibinfo{title}{A quantum
  teleportation inspired algorithm produces sentence meaning from word meaning
  and grammatical structure}}.
\newblock {\sl \bibinfo{journal}{Malaysian Journal of Mathematical Sciences}}
  \bibinfo{volume}{8}, pp. \bibinfo{pages}{15--25}.
\newblock \urlprefix\url{https://arxiv.org/abs/1305.0556}.

\bibitemdeclare{inproceedings}{Kindergarten}
\bibitem{Kindergarten}
\bibinfo{author}{Bob \surnamestart Coecke\surnameend} (\bibinfo{year}{2006}):
  \emph{\bibinfo{title}{Kindergarten quantum mechanics}}.
\newblock In: {\sl \bibinfo{booktitle}{AIP conference proceedings - Quantum
  Theory: Reconsiderations of the Foundations III}}, \bibinfo{volume}{810}, pp.
  \bibinfo{pages}{81--98}, \doi{10.1063/1.2158713}.
\newblock \urlprefix\url{https://arxiv.org/abs/quant-ph/0510032}.

\bibitemdeclare{article}{Coecke2015}
\bibitem{Coecke2015}
\bibinfo{author}{Bob \surnamestart Coecke\surnameend} \& \bibinfo{author}{Aleks
  \surnamestart Kissinger\surnameend} (\bibinfo{year}{2015}):
  \emph{\bibinfo{title}{{Categorical Quantum Mechanics I: Causal Quantum
  Processes}}}.
\newblock \urlprefix\url{http://arxiv.org/abs/1510.05468}.

\bibitemdeclare{article}{Coecke2016}
\bibitem{Coecke2016}
\bibinfo{author}{Bob \surnamestart Coecke\surnameend} \& \bibinfo{author}{Aleks
  \surnamestart Kissinger\surnameend} (\bibinfo{year}{2016}):
  \emph{\bibinfo{title}{{Categorical Quantum Mechanics II: Classical-Quantum
  Interaction}}}.
\newblock \doi{10.1142/S0219749916400207}.
\newblock \urlprefix\url{https://arxiv.org/abs/1605.08617}.

\bibitemdeclare{book}{CoeckeBOOK}
\bibitem{CoeckeBOOK}
\bibinfo{author}{Bob \surnamestart Coecke\surnameend} \& \bibinfo{author}{Aleks
  \surnamestart Kissinger\surnameend} (\bibinfo{year}{2017}):
  \emph{\bibinfo{title}{Picturing Quantum Processes: {A} First Course in
  Quantum Theory and Diagrammatic Reasoning}}.
\newblock \bibinfo{publisher}{Cambridge University Press},
  \doi{10.1017/9781316219317}.

\bibitemdeclare{inbook}{CoeckeMartin}
\bibitem{CoeckeMartin}
\bibinfo{author}{Bob \surnamestart Coecke\surnameend} \& \bibinfo{author}{Keye
  \surnamestart Martin\surnameend} (\bibinfo{year}{2011}):
  \emph{\bibinfo{title}{A Partial Order on Classical and Quantum States}}, pp.
  \bibinfo{pages}{593--683}.
\newblock \bibinfo{publisher}{Springer Berlin Heidelberg},
  \bibinfo{address}{Berlin, Heidelberg}, \doi{10.1007/978-3-642-12821-9\_10}.

\bibitemdeclare{inbook}{CPaqPav}
\bibitem{CPaqPav}
\bibinfo{author}{Bob \surnamestart Coecke\surnameend},
  \bibinfo{author}{{\'E}ric~Oliver \surnamestart Paquette\surnameend} \&
  \bibinfo{author}{Dusko \surnamestart Pavlovi{\'c}\surnameend}
  (\bibinfo{year}{2009}): \emph{\bibinfo{title}{Classical and Quantum
  Structuralism}}, p. \bibinfo{pages}{29–69}.
\newblock \bibinfo{publisher}{Cambridge University Press},
  \doi{10.1017/CBO9781139193313.003}.

\bibitemdeclare{article}{CPV}
\bibitem{CPV}
\bibinfo{author}{Bob \surnamestart Coecke\surnameend}, \bibinfo{author}{Dusko
  \surnamestart Pavlovi{\'c}\surnameend} \& \bibinfo{author}{Jamie
  \surnamestart Vicary\surnameend} (\bibinfo{year}{2013}):
  \emph{\bibinfo{title}{A new description of orthogonal bases}}.
\newblock {\sl \bibinfo{journal}{Mathematical Structures in Computer Science}}
  \bibinfo{volume}{23}, pp. \bibinfo{pages}{555--567},
  \doi{10.1017/S0960129512000047}.
\newblock \urlprefix\url{https://arxiv.org/abs/0810.0812}.

\bibitemdeclare{article}{Coeckea}
\bibitem{Coeckea}
\bibinfo{author}{Bob \surnamestart Coecke\surnameend},
  \bibinfo{author}{Mehrnoosh \surnamestart Sadrzadeh\surnameend} \&
  \bibinfo{author}{Stephen \surnamestart Clark\surnameend}
  (\bibinfo{year}{2011}): \emph{\bibinfo{title}{{Mathematical Foundations for a
  Compositional Distributional Model of Meaning}}}.
\newblock {\sl \bibinfo{journal}{Linguistic Analysis}}
  \bibinfo{volume}{36}(\bibinfo{number}{1-4}), pp. \bibinfo{pages}{345--384}.
\newblock \urlprefix\url{https://arxiv.org/abs/1003.4394}.

\bibitemdeclare{article}{Cunningham-privite}
\bibitem{Cunningham-privite}
\bibinfo{author}{Oscar \surnamestart Cunningham\surnameend}:
  \bibinfo{note}{Private communication}.

\bibitemdeclare{article}{Horodecki2007}
\bibitem{Horodecki2007}
\bibinfo{author}{Ryszard \surnamestart Horodecki\surnameend},
  \bibinfo{author}{Pawel \surnamestart Horodecki\surnameend},
  \bibinfo{author}{Michal \surnamestart Horodecki\surnameend} \&
  \bibinfo{author}{Karol \surnamestart Horodecki\surnameend}
  (\bibinfo{year}{2009}): \emph{\bibinfo{title}{{Quantum entanglement}}}.
\newblock {\sl \bibinfo{journal}{Reviews of Modern Physics}},
  \doi{10.1103/RevModPhys.81.865}.
\newblock \urlprefix\url{https://arxiv.org/abs/quant-ph/0702225}.

\bibitemdeclare{inproceedings}{Marsden15}
\bibitem{Marsden15}
\bibinfo{author}{Daniel \surnamestart Marsden\surnameend}
  (\bibinfo{year}{2015}): \emph{\bibinfo{title}{A Graph Theoretic Perspective
  on CPM(Rel)}}.
\newblock In: {\sl \bibinfo{booktitle}{Proceedings 12th International Workshop
  on Quantum Physics and Logic}}, pp. \bibinfo{pages}{273--284},
  \doi{10.4204/EPTCS.195.20}.

\bibitemdeclare{inproceedings}{Penrose1971}
\bibitem{Penrose1971}
\bibinfo{author}{Roger \surnamestart Penrose\surnameend}
  (\bibinfo{year}{1971}): \emph{\bibinfo{title}{{Applications of Negative
  Dimensional Tensors}}}.
\newblock In: {\sl \bibinfo{booktitle}{Combinatorial Mathematics and its
  Applications}}, \bibinfo{publisher}{Academic Press}, pp.
  \bibinfo{pages}{221--244}.

\bibitemdeclare{inproceedings}{Piedeleu2015}
\bibitem{Piedeleu2015}
\bibinfo{author}{Robin \surnamestart Piedeleu\surnameend},
  \bibinfo{author}{Dimitri \surnamestart Kartsaklis\surnameend},
  \bibinfo{author}{Bob \surnamestart Coecke\surnameend} \&
  \bibinfo{author}{Mehrnoosh \surnamestart Sadrzadeh\surnameend}:
  \emph{\bibinfo{title}{{Open System Categorical Quantum Semantics in Natural
  Language Processing}}}.
\newblock In: {\sl \bibinfo{booktitle}{Proceedings of the 6th Conference on
  Algebra and Coalgebra in Computer Science}},
  \doi{10.4230/LIPIcs.CALCO.2015.270}.
\newblock \urlprefix\url{http://arxiv.org/abs/1502.00831}.

\bibitemdeclare{article}{Selinger2007}
\bibitem{Selinger2007}
\bibinfo{author}{Peter \surnamestart Selinger\surnameend}
  (\bibinfo{year}{2007}): \emph{\bibinfo{title}{{Dagger Compact Closed
  Categories and Completely Positive Maps. (Extended Abstract)}}}.
\newblock {\sl \bibinfo{journal}{Electronic Notes in Theoretical Computer
  Science}} \bibinfo{volume}{170}, pp. \bibinfo{pages}{139--163},
  \doi{10.1016/j.entcs.2006.12.018}.

\bibitemdeclare{book}{VonNeumann1955}
\bibitem{VonNeumann1955}
\bibinfo{author}{John \surnamestart {Von Neumann}\surnameend}
  (\bibinfo{year}{1955}): \emph{\bibinfo{title}{{Mathematical Foundations of
  Quantum Mechanics}}}.
\newblock \doi{10.2307/2313034}.

\bibitemdeclare{inproceedings}{Weteringen}
\bibitem{Weteringen}
\bibinfo{author}{John \surnamestart van~de Wetering\surnameend}
  (\bibinfo{year}{2016}): \emph{\bibinfo{title}{Entailment Relations on
  Distributions}}.
\newblock In: {\sl \bibinfo{booktitle}{Proceedings of the 2016 Workshop on
  Semantic Spaces at the Intersection of NLP, Physics and Cognitive Science,
  SLPCS@QPL 2016, Glasgow, Scotland, 11th June 2016.}}, pp.
  \bibinfo{pages}{58--66}, \doi{10.4204/EPTCS.221.7}.

\bibitemdeclare{article}{Zwart2017}
\bibitem{Zwart2017}
\bibinfo{author}{Maaike \surnamestart Zwart\surnameend} \& \bibinfo{author}{Bob
  \surnamestart Coecke\surnameend} (\bibinfo{year}{2017}):
  \emph{\bibinfo{title}{{Double Dilation $\neq$ Double Mixing}}}.
\newblock \urlprefix\url{http://arxiv.org/abs/1704.02309}.

\end{thebibliography}
\end{document}